\documentclass[11pt,a4paper,english]{article}
\usepackage[T1]{fontenc}
\usepackage[utf8]{inputenc}
\usepackage{authblk}
\usepackage{verbatim}
\usepackage{authblk}
\usepackage{cancel}
\usepackage[square,numbers]{natbib}	  
\title{Einstein Equations for a Noncommutative Spacetime of Lie-Algebraic Type}
\author[1]{Albert Much\footnote{amuch@matmor.unam.mx}}
\author[2]{Marcos Rosenbaum \footnote{mrosen@nucleares.unam.mx}}
\author[3]{Jos\'e David  Vergara \footnote{vergara@nucleares.unam.mx}}
\author[4]{Diego Vidal-Cruzprieto \footnote{dvc501@york.ac.uk}}
\affil[1]{Centro de Ciencias Matem\'aticas\\
	UNAM\\
	Morelia, Michoac\'an, M\'exico.}  
\affil[2,3]{Instituto de Ciencias Nucleares, UNAM,  D.F., M\'exico. }  
\affil[4]{Department of Mathematics, The University of York, Heslington, York, YO10 5DD, UK.}    
\usepackage{float}            
\usepackage[T1]{fontenc}   
\usepackage[cmex10]{amsmath}  
\usepackage{amstext}          
\usepackage{mathrsfs}
\usepackage{amsfonts}         
\usepackage{amssymb}          
\usepackage{amsthm}   
\usepackage{bm}               
\usepackage{enumerate}        
\usepackage[bottom]{footmisc} 
\usepackage{array}            
\usepackage{textcomp}   
\usepackage[right]{eurosym}   
\usepackage[square,numbers]{natbib}	  

\usepackage{amsmath}  
\usepackage[colorlinks,pdfpagelabels,pdfstartview = FitH,bookmarksopen = true,bookmarksnumbered =
true,linkcolor = black,plainpages = false,hypertexnames = false,citecolor = black]{hyperref}

\allowdisplaybreaks

\newtheorem{innercustomgeneric}{\customgenericname}
\providecommand{\customgenericname}{}
\newcommand{\newcustomtheorem}[2]{%
  \newenvironment{#1}[1]
  {%
   \renewcommand\customgenericname{#2}%
   \renewcommand\theinnercustomgeneric{##1}%
   \innercustomgeneric
  }
  {\endinnercustomgeneric}
}

\newcustomtheorem{proofof}{\textsc{Proof of Proposition}}

\newtheorem{theorem}{\textsc{Theorem}}[section]

\newtheorem{assumption}{\textsc{Assumption}}[section]
\newtheorem{proposition}{\textsc{Proposition}}[section]

\newtheorem{defi}{\textsc{Definition}}[section]

\newtheorem{remark}{Remark}[section]

\usepackage{fancyhdr}
\numberwithin{equation}{section} 
\begin{document}
	\maketitle
	\abstract{A general formula is calculated   for the  connection of  a central metric w.r.t.\ a noncommutative spacetime  of     Lie-algebraic type. This is done  by using the framework of linear connections on central bi-modules. The general formula is further on used to calculate the corresponding  Riemann tensor and prove the corresponding Bianchi identities and certain symmetries that are essential   to obtain a symmetric and divergenceless Einstein Tensor. In particular, the obtained  Einstein Tensor is not equivalent to the sum of  the noncommutative Riemann tensor and scalar,  as in the commutative case, but in addition a traceless term appears.} 
	
	\tableofcontents	
	 \section{Introduction} A  physically well-motivated and mathematically rigorous approach towards a  theory of quantum gravity is the program initiated by noncommutative geometry.   	Noncommutative geometry is, in addition, to being a generalized formalism of the commutative framework, a  novel approach that arises from a combination of fundamental concepts of general relativity and quantum mechanics. The procedure of going into smaller scales leads us to the point where our description of spacetime as a continuum is neither physically nor mathematically well-defined. This is the subject of investigation and the result of the so-called "geometrical" measurement problem, \cite{DFR}.
The argument stems from a combination of the uncertainty principle and the Schwarzschild radius implying that the measurement of a spacetime point with arbitrary precision is not possible, since the energy involved in this process creates a black hole and thus spacetime, around the Planck length, does not have a continuous structure. The solution is to introduce a quantum spacetime, i.e. a   noncommutative geometry that replaces the continuous spacetime structure. In particular, the geometrical measurement problem and its solution are similar to solving the motion of the electron in the atom by using a quantized version of the phase space, i.e.\ by introducing a noncommutative structure.	Hence, the introduction of a noncommutative structure is from a physical point of view well motivated. 
\newline\newline
In general relativity , the geometric entity that  quantifies the geometric part is the Einstein tensor. Hence,  one line of investigation in noncommutative geometry is to calculate a respective tensor (fulfilling reasonable requirements) for a class of noncommutative spacetimes. A unique Einstein tensor for a class of noncommutative spacetimes would help towards gaining more insight into the quantum nature of spacetime  and ultimately  towards a theory of quantum gravity.	Hence, our primary goal is to construct  a symmetric and divergenceless  \textbf{Einstein tensor for a class of  noncommutative spacetimes} of Lie-algebraic type. Since the Einstein equations   connect the Einstein tensor, a geometric quantity, with the energy-momentum tensor we require symmetry and covariance. This requirement is a deduction of the assumption, implicitly made in this work  that the symmetries and properties of the energy-momentum tensor remain  unchanged by the noncommutativity of spacetime. This excludes, for the moment, certain quantum anomalies. 
\newline\newline
The main difference to our approach w.r.t.\ using  deformation quantization (see for example \cite{A1}, \cite{A2}, \cite{A3}, \cite{A4} and references therein) is the fact that we use the algebraic path given in \cite{koszul_1986}, \cite{DV2} \cite{C}, \cite{DV1},  and \cite{BM}. Moreover, we restrict the investigation  to the case of a noncommutative spacetime of a Lie-algebraic type. This means that we consider classical coordinates  as generators of an algebra $\mathfrak{A}$ whose commutator relations are of the form  (see  \cite[Equation 2]{W})
\begin{align*}
[x^\mu,x^\nu]=C^{\mu\nu}_{\;\;\lambda}x^\lambda.
\end{align*}
Our motivation to study this case  is due to a specific   example of a noncommutative spacetime  of Lie-algebraic type  (i.e.\ bicrossproduct model) given in  \cite{BM}, where   quantum corrections  to the "classical" Einstein tensor appear. 
\newline\newline
As in the standard approach to noncommutative geometry, we define the connection as a linear map that acts on one-forms in the following fashion, \cite{ Mo, landi_2014}:
\begin{align*}
\nabla&: \Omega^1\rightarrow \Omega^1\otimes_{\mathfrak{A}}\Omega^1 \\
\nabla(a\omega)&=da\otimes_{\mathfrak{A}}\omega+a\nabla(\omega)\\
\nabla (\omega a)&=(\nabla \omega)a+\sigma(\omega\otimes_{\mathfrak{A}} da),
\end{align*}
where $a\in \mathfrak{A}$, $\omega \in \mathcal{E}$ is the module and $\sigma$ is a bimodular map known as the generalized braiding defined by,
\begin{align*}
\sigma(\omega\otimes_{\mathfrak{A}} da)&=\nabla (\omega a)-(\nabla \omega)a=\nabla ([\omega, a])-[(\nabla \omega),a]+\nabla(a\omega)-a\nabla(\omega)\\
&=da\otimes_{\mathfrak{A}}\omega+[a,(\nabla \omega)]+\nabla ([\omega, a]).
\end{align*} 
To achieve our main goal, we  use   one essential requirement: The line element has to be in the center of the algebra. Up to date,   central bi-modules (see \cite[Chapter 6-9]{landi_2014}, \cite{DV1}, \cite{DV2} and references therein) play a fundamental role in formulating geometric quantities in   noncommutative geometry (\textbf{NCG}). Moreover,  	the centrality of the metric tensor is due to keeping some of its classical tensorial features, and in particular, it allows us to invert it without any ambiguity, see \cite{BM}. By the use of the central metric, the covariant derivative,  and by imposing metricity and torsionlessness we get an explicit   expression for the connection symbol. The use of this symbol enables us  to  calculate the   Riemann tensor, the corresponding identities and symmetries thereof.  Equipped with these entities and symmetries, we construct  a divergenceless and symmetric Einstein tensor.  	\newline\newline
The organization of this paper is as follows; The second Section sets the precise mathematical framework in order, i.e. it defines the algebra, the corresponding universal differential algebra and the respective derivatives. The third and fourth Sections derive the most general formula for  equivalents of the geometric entities, such as the metric, the connection symbol and the Riemann tensor. We use these entities and their corresponding symmetries   in the fifth Section to define a unique Einstein tensor w.r.t.  a noncommutative spacetime  of   Lie-algebraic type. 
\section{Preliminaries}
In this Section we first define the Lie algebra, which plays a fundamental role in this work. Next, we construct the universal differential calculus, that  provides a generalization of the notions   used in commutative differential geometry. we follow the universal differential calculus developed by Dubois-Violette (c.f. \cite{DV1,DV2}). Finally, we introduce the centrality condition, used in \cite{BM},  that provides both \textit{tensor} features for the metric and, --alongside the universal differential calculus -- a way to calculate noncommutative corrections for the curvature quantities.   For the forthcoming  calculations, we use the fact that we have an associative and unital algebra induced from taking the largest embedding of the Lie algebra into an associative algebra, i.e. the universal enveloping. Equipped with the universal enveloping algebra, the universal differential calculus allows to define an algebra of forms, similar to the differential forms of the commutative case. Yet, the idea in noncommutative geometry is to use the recently introduced algebra and its derivations as opposed to the commutative case, where the algebra of smooth functions is used and, instead of derivations we have smooth vector fields. 
	\subsection{Differential Calculus} Let $\mathcal{A}$ be a Lie-algebra and let $\mathcal{U}(\mathcal{A})$,
	which we denote by $\mathfrak{A}$, be its universal enveloping algebra
	that has  as generators denoted by $x^\mu$,  fulfilling  the following commutation
	relations, (see  \cite[Equation 2]{W})
		\begin{align}
		[x^\mu,x^\nu]=C^{\mu\nu}_{\;\;\;\lambda}x^\lambda. \label{st}
		\end{align}	with structure constants $C^{\mu\nu}_{\;\;\;\lambda}\in\mathbb{C}$ for each $\mu, \nu$ and $\lambda$. In addition,  the Jacobi identities for the structure constants are imposed and they read,
		\begin{align*}
		C^{\mu\nu}_{\;\;\;\rho}C^{\lambda\rho}_{\;\;\;\sigma}+C^{\lambda\mu}_{\;\;\;\rho}C^{\nu\rho}_{\;\;\;\sigma}+C^{\nu\lambda}_{\;\;\;\rho}C^{\mu\rho}_{\;\;\;\sigma}=0.
		\end{align*}	 
		In this way the universal enveloping algebra  $\mathcal{U}(\mathcal{A})$,
		defines our spacetime. 		This definition serves as a starting point for defining a differential calculus and we complement it with the notion of a universal differential algebra. 
		 	
\begin{defi}
		Consider an associative algebra $\mathfrak{A}$ over $\mathbb{C}$ with unit. The \textbf{universal differential algebra} of forms (c.f. \cite[Chapter 7, Section 1]{landi_2014} and \cite[Chapter 3, Section 1]{C}) denoted by $\Omega(\mathfrak{A})=\bigoplus_{p}\Omega^p(\mathfrak{A})$,  is defined as follows:
		
		For $p = 0$ it is the algebra itself, i.e. $\Omega^0(\mathfrak{A}) = \mathfrak{A}$. The space $\Omega^1(\mathfrak{A})$ of one-forms is generated,
		as a left $\mathfrak{A}$-module by a $\mathbb{C}$-linear operator $d: \mathfrak{A} \to \Omega^1(\mathfrak{A})$, called the universal	differential or Karoubi operator, which satisfies the relations,
		\begin{equation}
	d^2=0, \;\;\;d(ab)=(da)b+adb, \;\;\; \forall a,b \in \mathfrak{A}. \label{univ_dif}
		\end{equation}
		If $\Omega^1(\mathfrak{A})$ is a left (right) $\mathfrak{A}$-module we can induce a right (left) $\mathfrak{A}$-module structure via the universal differential given in Equation (\ref{univ_dif}), which makes $\Omega^1(\mathfrak{A})$ a bi-module. The $\Omega^p(\mathfrak{A})$-space is build by 
		\begin{align*}
		\Omega^p(\mathfrak{A})=\underbrace{\Omega^1(\mathfrak{A})\otimes_{\mathfrak{A}}\cdots\otimes_{\mathfrak{A}}\Omega^1(\mathfrak{A})}_{p-\text{times}}.
		\end{align*}
		The product between any two one-forms is defined by juxtaposition
		\begin{align*}
			(x_0dx_1)(y_0dy_1)=x_0(dx_1)y_0dy_1=x_0d(x_1y_0)dy_1-x_0x_1dy_0dy_1,
		\end{align*}
		where the result is  generalized to the product of any $p$-form with any $(q-1)$-form,
		\begin{align*}
			(x_0dx_1\cdots &dx_p)(x_{p+1}dx_{p+2}\cdots dx_{p+q})=x_0dx_1\cdots (dx_p)x_{p+1}dx_{p+2}\cdots dx_{p+q}\\
			&=(-1)^px_0x_1dx_2\cdots dx_{p+q}+\sum_{i=1}^{p}x_0dx_1\cdots dx_{i-1}d(x_ix_{i+1})dx_{i+2}\cdots dx_{p+q}.
		\end{align*}
		An immediate consequence of the definition is that the differential algebra of forms is graded, since now one can extend $d$ as a linear operator $:\Omega^p(\mathfrak{A})\to\Omega^{p+1}(\mathfrak{A})$ that acts as follows
			\begin{align*}
			d(x_0dx_1\cdots dx_p)=:dx_0dx_1\cdots dx_p.
			\end{align*} 
			For any integer $p$ the following relations hold
			\begin{equation*}
			d^2=0, \;\;\;d(ab)=(da)b+(-1)^pa\,db, \;\;\;
			\end{equation*}
			where $a\in \Omega^p(\mathfrak{A})$, $b\in\Omega^1(\mathfrak{A})$.
		
\end{defi}

By recalling the fact that the universal algebra ${\mathfrak A}$ inherits the Lie algebraic commutators of  (\ref{st}) we obtain our first result. 
	\begin{proposition}	 The commutator of the differential of the algebra, i.e.\ $dx\in\Omega^1(\mathfrak A)$, with generators of the algebra ${\mathfrak A}$ has the following solution,
		\begin{equation} \label{defo}
		[dx^\mu,x^\nu]=\left(\frac{1}{2}C^{\mu\nu}_{\;\;\;\lambda}+S^{\mu\nu}_{\;\;\;\lambda}\right)dx^\lambda=:D^{\mu\nu}_{\;\;\;\lambda}dx^\lambda ,
		\end{equation}
		where the   tensor components $S^{\mu\nu}_{\;\;\;\lambda}\in\mathbb{C}$ are symmetric in $\mu,\nu$.
		\begin{proof}
			We act with the universal differential on the commutator, where the left hand side renders commutators with each of them containing a one-form basis and a generator, while the right hand side is just the structure constants of the Lie-algebra contracted with a one form-basis,
			\begin{align*}
			d[x^\mu,x^\nu]&=[dx^\mu,x^\nu]+[x^\mu,dx^\nu]\\&=\left(\frac{1}{2}C^{\mu\nu}_{\;\;\;\lambda}+S^{\mu\nu}_{\;\;\;\lambda}\right)dx^\lambda-\left(\frac{1}{2}C^{\nu\mu}_{\;\;\;\lambda}+S^{\nu\mu}_{\;\;\;\lambda}\right)dx^\lambda\\&=C^{\mu\nu}_{\;\;\;\lambda}dx^\lambda.
			\end{align*}
			 
		\end{proof}
	\end{proposition} 
	Precisely  the notion of a universal differential algebra  allows us to regard $dx^\mu$ as a one-form basis for the space $\Omega^1( {\mathfrak A})$.

	\begin{defi}\label{def:der} \textbf{Derivations} $\hat{\partial}: \mathfrak{A}\mapsto \mathfrak{A}$ are  maps on the universal enveloping algebra   satisfying the following commutation relation  \cite[Equation 5.5]{AS09},
\begin{align}
 [\hat{\partial}_{\alpha},x^{\beta}] =\delta_{\alpha}^{\beta}+D_{\ \ \alpha}^{\beta\gamma}\hat{\partial}_{\gamma},\label{der}
  \end{align} 
    with $D_{\ \ \alpha}^{\beta\gamma}$ given in (\ref{defo}).
   \end{defi}   
    \begin{proposition}
The  Jacobi identities, obtained from derivations and generators of the algebra $\mathfrak{A}$ imply that the  derivatives commute, i.e. $[\hat{\partial}_{\nu},\hat{\partial}_{\alpha}]=0$, and that the  structure constants $D$ and $C$ obey the following relations
\begin{align*}
D^{\nu\rho}_{\ \ \alpha}D_{\ \ \rho}^{\mu\beta}+ D^{\mu\rho}_{\ \ \alpha}D_{ \ \ \rho}^{\nu\beta} +C^{\mu\nu}_{\ \ \rho}D_{ \ \ \alpha}^{\rho\beta}=0.
\end{align*}
\begin{proof} 			
 The first identity is obtained by combining two derivatives and the generator of the algebra and by using the relation  $[\hat{\partial}_{\nu},\hat{\partial}_{\alpha}]=:F_{\nu\alpha}^{\gamma}\hat{\partial}_{\gamma}$, due to  \cite[Section 3.2]{DuboisViolette:1995hh}, 	where $F\in\mathbb{C}$, i.e.\
\begin{align*}0&=
[x^{\mu},[\hat{\partial}_{\nu},\hat{\partial}_{\alpha}]]+[\hat{\partial}_{\nu},[\hat{\partial}_{\alpha},x^{\mu}]]- [\hat{\partial}_{\alpha},[\hat{\partial}_{\nu},x^{\mu}]]
\\&=
 F_{\nu\alpha}^{\gamma}[x^{\mu},\hat{\partial}_{\gamma}]
+\left( D_{\ \ \alpha}^{\mu\gamma}
 F_{\nu\gamma}^{\beta}\hat{\partial}_{\beta} - 
\mu\leftrightarrow\nu\right)  \\&=F_{\nu\alpha}^{\gamma}
\left(
\delta_{\gamma}^{\mu}+D_{\ \ \gamma}^{\mu\rho}\hat{\partial}_{\rho}
\right) 
+\left( D_{\ \ \alpha}^{\mu\gamma}
F_{\nu\gamma}^{\beta}\hat{\partial}_{\beta} - 
\mu\leftrightarrow\nu\right) \\&=F_{\nu\alpha}^{\mu}+\left(F_{\nu\alpha}^{\gamma}D_{\ \ \gamma}^{\mu\rho}
+  D_{\ \ \alpha}^{\mu\gamma}
F_{\nu\gamma}^{\rho} -   D_{\ \ \alpha}^{\nu\gamma}
F_{\mu\gamma}^{\rho}\right)\hat{\partial}_{\rho}=0,
 \end{align*}
from which  $F_{\nu\gamma}^{\beta}=0$ follows. Next, we take the Jacobi identity of the derivative with two elements of the algebra, i.e.\
\begin{align*}0&=
[x^{\mu},[x^{\nu},\hat{\partial}_{\alpha}]]-[x^{\nu},[x^{\mu},\hat{\partial}_{\alpha}]]+ [\hat{\partial}_{\alpha},[x^{\mu},x^{\nu}]]\\&=D^{\nu\rho}_{\ \ \alpha}
[x^{\mu},\hat{\partial}_{\rho}] - 
\mu\leftrightarrow\nu+C^{\mu\nu}_{\ \ \rho} [\hat{\partial}_{\alpha}, x^{\rho} ]
\\&=D^{\nu\rho}_{\ \ \alpha}
\left(
\delta_{\rho}^{\mu}+D_{\ \ \rho}^{\mu\beta}\hat{\partial}_{\beta}
\right)   - 
\mu\leftrightarrow\nu+C^{\mu\nu}_{\ \ \rho}
\left(
\delta_{\alpha}^{ \rho}+D_{\ \ \alpha}^{\rho\beta}\hat{\partial}_{\beta}
\right) 
\\&= D^{\nu\mu}_{\ \ \alpha}
+ D^{\nu\rho}_{\ \ \alpha}D_{\ \ \rho}^{\mu\beta}\hat{\partial}_{\beta}
   -  D^{\mu\nu}_{\ \ \alpha}
   + D^{\mu\rho}_{\ \ \alpha}D_{\ \ \rho}^{\nu\beta}\hat{\partial}_{\beta} 
   +C^{\mu\nu}_{\ \ \alpha}
   +C^{\mu\nu}_{\ \ \rho}D_{\ \ \alpha}^{\rho\beta}\hat{\partial}_{\beta} 
\\&= D^{\nu\mu}_{\ \ \alpha}-  D^{\mu\nu}_{\ \ \alpha}  +C^{\mu\nu}_{\ \ \alpha}
+ (D^{\nu\rho}_{\ \ \alpha}D_{\ \ \rho}^{\mu\beta}+ D^{\mu\rho}_{\ \ \alpha}D_{ \ \ \rho}^{\nu\beta} +C^{\mu\nu}_{\ \ \rho}D_{ \ \ \alpha}^{\rho\beta})
 \hat{\partial}_{\beta} =0.
\end{align*}
The term without derivative is zero since the skew-symmetric term of  the constants $D$ are $1/2C$. 
\end{proof}
 \end{proposition}
\begin{remark}[Leading order]
	W.r.t. the subsequent discussion all equalities are to first order in the structure constants $D^{\mu\nu}_{\ \ \lambda}$. 
\end{remark}
  \begin{proposition}
 	Let us define derivations $D_{\alpha}:\mathfrak{A}\mapsto \mathfrak{A}$, by   using the  derivation $\hat{\partial}$  given in Definition \ref{def:der},
 	\begin{align}
 	D_{\alpha} A :=   [\hat{\partial}_{\alpha},A ],  \label{der1}
 	\end{align} 
 	for all $\alpha$ and $A\in \mathfrak{A}$. Then,  the exterior derivative $d$ is given as
 	\begin{align*}
 	dA=(D_{\alpha} A)\,dx^{\alpha}.
 	\end{align*}  
 \end{proposition}

 \begin{proof} Let us take an object of the form $A^{\mu\nu}:=x^\mu x^\nu\in \mathfrak{A}$. By applying the exterior derivative (also known as the Karoubi operator) $d$ on this object it is possible to extract the derivation, i.e. $dA^{\mu\nu}= (D_{\alpha}A^{\mu\nu})dx^{\alpha}$ by using the Leibniz  rule, i.e.\
 	\begin{align*}
 	dA^{\mu\nu}&=d(x^{\mu}x^{\nu})=(dx^{\mu})x^{\nu}+x^{\mu}dx^{\nu}\\&=[dx^{\mu},x^{\nu}]+x^{\mu}dx^{\nu}+x^{\nu}dx^{\mu}
 	\\&=\underbrace{\left(D^{\mu\nu}_{\ \ \rho}+x^{\mu}\delta^{\nu}_{\rho}+x^{\nu}\delta^{\mu}_{\rho}\right)}_{=:D_{{\rho}}A^{\mu\nu}}dx^{\rho}.
 	\end{align*}
 	Next, we use   Equation (\ref{der}) and the defining Equation (\ref{der1}),
 	\begin{align*}
 	D_{\rho}A^{\mu\nu}&=[\hat{\partial}_{\rho},x^{\mu}x^{\nu}]\,1_{\mathfrak{A}}
 	\\&=[\hat{\partial}_{\rho},x^{\mu}]x^{\nu}\,1_{\mathfrak{A}}
 	+x^{\mu}[\hat{\partial}_{\rho},x^{\nu}]\,1_{\mathfrak{A}}	\\&=[\hat{\partial}_{\rho},x^{\mu}]x^{\nu}\,1_{\mathfrak{A}}
 	+x^{\mu}[\hat{\partial}_{\rho},x^{\nu}]\,1_{\mathfrak{A}}
 	\\&=
 	\delta_{\rho}^{\mu}x^{\nu}\,1_{\mathfrak{A}}+D_{\ \ \rho}^{\mu\gamma}\hat{\partial}_{\gamma}x^{\nu}\,1_{\mathfrak{A}}+
 	x^{\mu} 	\delta_{\rho}^{\nu}\,1_{\mathfrak{A}}+	x^{\mu} D_{\ \ \rho}^{\nu\gamma}\hat{\partial}_{\gamma}\,1_{\mathfrak{A}}	\\&=
 	\delta_{\rho}^{\mu}x^{\nu}\,1_{\mathfrak{A}}+	x^{\mu} 	\delta_{\rho}^{\nu}\,1_{\mathfrak{A}}+
 	D_{\ \ \rho}^{\mu\gamma}[\hat{\partial}_{\gamma},x^{\nu}] \,1_{\mathfrak{A}}
 	+	x^{\mu} D_{\ \ \rho}^{\nu\gamma}\hat{\partial}_{\gamma}\,1_{\mathfrak{A}}
 	+	x^{\nu} D_{\ \ \rho}^{\mu\gamma}\hat{\partial}_{\gamma}\,1_{\mathfrak{A}}
 	\\&=\left(  
 	\delta_{\rho}^{\mu}x^{\nu} +	x^{\mu} 	\delta_{\rho}^{\nu} +
 	D_{\ \ \rho}^{\mu\nu}\right)\,1_{\mathfrak{A}}, 
 	\end{align*}
 	The result agrees with the one obtained by making use of the Leibniz rule and is by linearity extendable to any polynomial.
 \end{proof} 
	\subsection{Centrality condition}
In addition to  the differential calculus, the centrality (w.r.t. the algebra) of the  metric is an essential requirement, since it allows us to keep classical tensor-like features for the metric, such as a well-defined inverse. In this section, we introduce and investigate the implications of the centrality requirement. We begin by defining the line element as a tensor product of two one-forms. 
		\begin{defi}
			The \textbf{center of the algebra of two-forms},  denoted by $\mathcal{Z}(\Omega^1({\mathfrak A} )\otimes_{{\mathfrak A} }\Omega^1({\mathfrak A} ))$ is defined as the following set
			\begin{equation*}
			\mathcal{Z}(\Omega^1({\mathfrak A} )\otimes_{{\mathfrak A} }\Omega^1({\mathfrak A} )):=\left\{
			z\in \Omega^1({\mathfrak A} )\otimes_{{\mathfrak A} }\Omega^1({\mathfrak A} )\;|\;[z,a]=0,\,\,\forall a\in {\mathfrak A}   \right\}.
			\end{equation*}
		\end{defi}
	\begin{defi}
		Let the \textbf{line element} $g$ be an element of the center of two-forms, i.e. $g\in\mathcal{Z}(\Omega^1({\mathfrak A} )\otimes_{{\mathfrak A} }\Omega^1({\mathfrak A} ))$, expressed in terms of the basis as,
		\begin{align*}
		g=g_{\mu\nu}dx^\mu\otimes_{{\mathfrak A} }dx^\nu.
		\end{align*} 
	\end{defi}
		\begin{assumption}
				In the following we assume symmetry for the metric components $$g_{\mu\nu}=g_{\nu\mu}.$$
		\end{assumption} 
	\begin{proposition} 
	The commutator of the generators of the algebra ${\mathfrak A}$ and the central line element $g\in\mathcal{Z}(\Omega^1({\mathfrak A} )\otimes_{{\mathfrak A} }\Omega^1({\mathfrak A} ))$ is given in components by,
		\begin{equation}
		[x^\lambda, g_{\mu\nu}]=D^{\alpha\lambda}_{\;\;\;\mu}g_{\alpha\nu}+D^{\alpha\lambda}_{\;\;\;\nu}g_{\alpha\mu}. \label{cen}
		\end{equation}
		\begin{proof}
			Demanding\footnote{For the rest of this text we   omit the subscript on the tensor product.}  centrality means that   equation $[x^\lambda , g]=0$ holds for all values of $\lambda$, thus
			\begin{align*}
			0&=[x^\lambda ,g ]=[x^\lambda, g_{\mu\nu}dx^\mu\otimes dx^\nu]\\&=[x^\lambda, g_{\mu\nu}]dx^\mu\otimes dx^\nu+ g_{\mu\nu}[x^\lambda,dx^\mu\otimes dx^\nu] \nonumber \\
			&=[x^\lambda, g_{\mu\nu}]dx^\mu\otimes dx^\nu-D^{\mu\lambda}_{\;\;\alpha} g_{\mu\nu}dx^\alpha\otimes dx^\nu-D^{\nu\lambda}_{\;\;\alpha}g_{\mu\nu}dx^\mu\otimes dx^\alpha,  
			\end{align*}
			where in the last lines we used the Leibniz rule and the solution of the commutator relation between the algebra and the differentials given in Equation (\ref{defo}).
		\end{proof}
	\end{proposition}
By using the explicit result of the centrality condition and the algebra of the outer derivatives we calculate some commutators between the metric and its derivations with the generators of the algebra that will be of use later on.
 \begin{proposition}[]\label{prop:cmd}
The commutator relations of the generators of the algebra $\mathfrak{A}$ with the inverse metric and with   derivations of the metric are given by, 
	\begin{align}\nonumber
	[x^{\beta},g^{\sigma\nu}]
	&=-
	(  D^{\sigma\beta}_{\ \ \rho} g^{\rho\nu}+
	D^{\nu\beta}_{\ \ \rho}g^{\rho\sigma}
	),
\\\label{e2}
	[x^{\mu},D_{ \alpha}g_{\lambda\nu} ]&=
	D_{\ \ \lambda}^{\beta\mu}D_{ \alpha}g_{\beta\nu}+  D_{\ \ \nu}^{\beta\mu}D_{ \alpha}g_{\beta\lambda}-D_{ \ \ \alpha}^{\mu\beta}
	D_{\beta} g_{\lambda\nu}.
	\end{align}  
	\begin{proof}
See Appendix \ref{proof:e2}.

\end{proof}
\end{proposition}

\section{Quantum connection}
 The concept of the connection for noncommutative algebras is introduced in this section.
 The connection plays a central role both in commutative and non-commutative geometry since it encodes the information needed to calculate the curvature tensors and scalars.   First, we  define the covariant derivative for left $a\omega$ and right $\omega a$ modules, where $a$ and $\omega$ are elements of the algebra $\mathfrak{A}$ and the module $\Omega^1( \mathfrak{A})$ respectively.  By using the central metric and the covariant derivative we impose metricity and torsionlessness in order to obtain an explicit connection symbol.   With regards to  the connection see \cite{koszul_1986}, see \cite{DV2} and   \cite{DV1}.  \begin{defi}\label{c1}
 	The \textbf{connection} $\nabla:\Omega^1(\mathfrak{A})\rightarrow \Omega^1(\mathfrak{A})\otimes_{\mathfrak{A}}\Omega^1(\mathfrak{A})$ is a linear map that acts on one-forms in the following fashion
 	\begin{align*} 
 	&\nabla(a\omega)=da\otimes_{\mathfrak{A}}\omega+a\nabla(\omega),\\
 	&\nabla (\omega a)=(\nabla \omega)a+\sigma(\omega\otimes_{\mathfrak{A}} da),
 	\end{align*}
 	where $a\in \mathfrak{A}$, $\omega \in \Omega^1( \mathfrak{A})$ is the module  and the symbol $\sigma$ is a bi-modular map known as generalized braiding. The bi-modular map, also called the \textbf{generalized braiding}		$\sigma$ is obtained by using both  expressions of  covariant derivatives w.r.t. a left or right module (see Definition \ref{c1}), i.e.\
 	\begin{align}
 	\sigma(\omega\otimes_{\mathfrak{A}} da)&=\nabla (\omega a)-(\nabla \omega)a=\nabla ([\omega, a])-[(\nabla \omega),a]+\nabla(a\omega)-a\nabla(\omega)\nonumber\\
 	&=da\otimes_{\mathfrak{A}}\omega+[a,(\nabla \omega)]+\nabla ([\omega, a]). \label{braid}
 	\end{align}
 	Moreover, by using the generalized braiding the covariant derivative   applied to the tensor product over $\mathfrak{A}$ of $s$ copies of  $\Omega^1(\mathfrak{A})$, as in \cite{Mo},   is explicitly given by
 	\begin{align*}
 	\nabla(\omega\otimes_{\mathfrak{A}}\omega')
 	=\nabla \omega\otimes_{\mathfrak{A}}\omega'+
 	\sigma_s( \omega\otimes_{\mathfrak{A}}\nabla\omega'),
 	\end{align*}
 	where $\omega$ is in $\Omega^1(\mathfrak{A})$
 	and $\omega'$ is in $\otimes_{\mathfrak{A}}^{s-1}\Omega^{1}(\mathfrak{A})$
 	and the generalized braiding with sub-index $s$   is given by the relation 
 	\begin{align*}
 	\sigma_s=\sigma \otimes_{\mathfrak{A}}\underbrace{
 		\mathbb{I}\otimes_{\mathfrak{A}}\mathbb{I}\otimes_{\mathfrak{A}}\cdots\otimes_{\mathfrak{A}}\mathbb{I}}_{s-1\,\,\text{times}}. 
 	\end{align*}
 \end{defi}  Next, we extend this definition to the tensor product over the algebra of an arbitrary number of copies of the module (see \cite[Eq 3.9]{Mo}), which we use in Section \ref{sec:4}  to calculate the Riemann tensor. 
 \begin{defi}\label{ext}
 	Let $\Omega^{*}(\mathfrak{A})=\bigoplus_{p=0}^n\Omega^p(\mathfrak{A})$ for an arbitrary integer $n$. Then, the following mapping 
 	\begin{align*}
 	D:\Omega^{*}(\mathfrak{A})\otimes_{\mathfrak{A}}\Omega^1(\mathfrak{A})\to\Omega^{*}(\mathfrak{A})\otimes_{\mathfrak{A}}\Omega^1(\mathfrak{A}),
 	\end{align*}  
 	extends the notion of covariant derivative beyond one-forms and acts as
 	\begin{align*}
 	D(\omega\otimes \omega')&=d(\omega)\otimes \omega'+(-1)^{\text{deg}\omega}\omega\cdot\nabla(\omega'),\\
 	\omega\cdot(\eta\otimes\eta')&:=(\omega\wedge\eta)\otimes\eta',
 	\end{align*}
 	with $\omega\in \Omega^{*}(\mathfrak{A})$,  $\omega'\in\Omega^1(\mathfrak{A})$ and {wedge product}    defined as
 	\begin{align*}
 	\omega\wedge\omega'=\omega\otimes \omega'-\omega'\otimes \omega.
 	\end{align*} 
 \end{defi} 
\begin{remark}For the covariant derivative of a one-form,  generated by applying the Karoubi operator on generators of the algebra $\mathcal{A}$, we write 
	\begin{align*}
	\nabla(dx^\mu)=&-  \tilde{\Gamma}^{\mu}_{\rho\sigma}  dx^\rho\otimes dx^\sigma,
	\end{align*}
	where $\tilde{\Gamma}$ is the connection symbol up to first order. For the connection symbol up to zero order we use the symbol ${\Gamma}$.
\end{remark}
Next,   we derive  the most basic formula for our present work; the covariant derivative up to first order in the noncommutative parameters. This formula   is of use to extend   the Riemann curvature to the noncommutative framework. First, we calculate torsion, which is given by $T=d-\wedge\circ\nabla$, \cite{DV1}. By using the condition of zero torsion we advance to the metricity condition from   which we extract, by standard methods, the connection symbol. 
\begin{proposition}\label{prop:tors}
	Let the one-forms of the universal differential algebra $\Omega^1(\mathcal{A})$ be given by a basis of the form $\omega^\mu=dx^\mu$. Then, the torsion  in components reads  	
$$T^\mu_{\alpha\beta}=\tilde{\Gamma}^\mu_{\alpha\beta}-\tilde{\Gamma}^\mu_{\beta\alpha}.$$
	\begin{proof}
		In the case at hand  we have $\omega^\mu=dx^\mu$, therefore
		\begin{align*}
		T(dx^\mu)&=d^2x^\mu-\wedge\circ\nabla(dx^\mu)=\wedge(\tilde{\Gamma}^\mu_{\alpha\beta}dx^\alpha\otimes dx^\beta)=\tilde{\Gamma}^\mu_{\alpha\beta}dx^\alpha\wedge dx^\beta.
		\end{align*} 
	\end{proof}	 
	\begin{assumption}	\label{ass:tor} 
		Henceforth, we  assume \textbf{zero torsion}, i.e. 
		\begin{align*}
		T^\mu_{\alpha\beta}=0.
		\end{align*}
	\end{assumption}
\end{proposition}
 
	\begin{theorem}\label{P3.1}[\textbf{General formula}]
			By imposing torsionlessness (see Assumption \ref{ass:tor}) and the metricity condition, i.e.\ $\nabla g=0$,		the \textbf{connection symbol}   for the most general Lie-algebraic type of noncommutative spacetime is explicitly given by 
		\begin{align*}
		\tilde{\Gamma}^{\kappa}_{\alpha\mu} & 
		= \overline {\Gamma}_{\alpha\mu}^{\kappa}-g^{\kappa\nu}  \Sigma_{\alpha\mu\nu},
		\end{align*}
		with the following defined symbols 
		\begin{align*}
		\overline {\Gamma}_{\alpha\mu}^{\kappa}
		&:=\frac{1}{2}g^{\kappa\nu}\left(D_{ {\alpha}}g_{\mu\nu }+D_{ {\mu}}g_{\nu\alpha}-D_{ {\nu}}g_{\alpha\mu} \right), \;\;\;\;\;\;\;\;\;\;
		\\\Sigma_{\alpha\mu\nu}&:=g_{\lambda\beta}\Sigma^{\lambda\beta}_{\alpha\mu\nu} =-g_{\lambda\beta}\overline {\Gamma}^\beta_{\rho\nu}([x^\rho,\overline {\Gamma}^\lambda_{\alpha\mu}]+D^{\lambda\rho}_{\ \ \sigma}\overline {\Gamma}^{\sigma}_{\alpha\mu}-D^{\sigma\rho}_{\ \ \alpha}\overline {\Gamma}^\lambda_{\sigma\mu}-D^{\sigma\rho}_{\ \ \mu}\overline {\Gamma}^\lambda_{\alpha\sigma}).
		\end{align*}  
	\end{theorem}

\begin{proof}
	In order to compare we rewrite the derivative of an arbitrary $n$-form (zero forms included) as $dg=(D_{\alpha}g)\,dx^{\alpha}$.  		Next, we include the first term of the covariant derivative and end up with
	\begin{align*}
	0&=\underbrace{dg_{\mu\nu}}_{(D_{ {\alpha}}g_{\mu\nu})\,dx^{\alpha}} \otimes dx^\mu\otimes dx^\nu-g_{\mu\nu}\tilde{\Gamma}^\mu_{\alpha\beta} dx^\alpha\otimes dx^\beta \otimes dx^\nu-g_{\mu\nu}\tilde{\Gamma}^\nu_{\alpha\beta}dx^\alpha\otimes dx^\mu\otimes dx^\beta\\
	&+g_{\mu\nu}\Gamma^\nu_{\alpha\beta}([x^\alpha,\Gamma^\mu_{\rho\sigma}]+D^{\mu\alpha}_{\ \ \eta}\Gamma^\eta_{\rho\sigma}-D^{\eta\alpha}_{\ \ \rho}\Gamma^\mu_{\eta\sigma}-D^{\eta\alpha}_{\ \ \sigma}\Gamma^\mu_{\rho\eta})dx^\rho\otimes dx^\sigma\otimes dx^\beta\\
	&-g_{\mu\nu}\Sigma^{\mu\nu}_{\alpha\beta\lambda}dx^\alpha\otimes dx^\lambda\otimes dx^\beta\\
	0&=(D_{ {\alpha}}g_{\mu\nu }-g_{\beta\nu}\tilde{\Gamma}^\beta_{\alpha\mu}-g_{\mu\beta}\tilde{\Gamma}^\beta_{\alpha\nu}-g_{\lambda\beta}\Sigma^{\lambda\beta}_{\alpha\nu\mu})dx^\alpha\otimes dx^\mu\otimes dx^\nu\\
	&+g_{\sigma\beta}\Gamma^\beta_{\rho\nu}([x^\rho,\Gamma^\sigma_{\alpha\mu}]+D^{\sigma\rho}_{\ \ \eta}\Gamma^\eta_{\alpha\mu}-D^{\eta\rho}_{\ \ \alpha}\Gamma^\sigma_{\eta\mu}-D^{\eta\rho}_{\ \ \mu}\Gamma^\sigma_{\alpha\eta})dx^\alpha\otimes dx^\mu\otimes dx^\nu,
	\end{align*}
	where we defined the term $\Sigma^{\lambda\beta}_{\alpha\nu\mu}dx^\mu:=[dx^\lambda,\Gamma^\beta_{\alpha\nu }]$ and 	in the second equality we just grouped terms together. Hence, for the coefficients we end up with 
	\begin{align*}
	0&=D_{ {\alpha}}g_{\mu\nu}-g_{\beta\nu}\tilde{\Gamma}^\beta_{\alpha\mu}-g_{\mu\beta}\tilde{\Gamma}^\beta_{\alpha\nu}-g_{\lambda\beta}\Sigma^{\lambda\beta}_{\alpha\nu\mu}\\
	&\qquad \qquad+g_{\lambda\beta}\Gamma^\beta_{\rho\nu}([x^\rho,\Gamma^\lambda_{\alpha\mu}]+D^{\lambda\rho}_{\ \ \sigma}\Gamma^\sigma_{\alpha\mu}-D^{\sigma\rho}_{\ \ \alpha}\Gamma^\lambda_{\sigma\mu}-D^{\sigma\rho}_{\ \ \mu}\Gamma^\lambda_{\alpha\sigma}).
	\end{align*}
	Before cycling the indices in order to obtain the quantum corrected connection symbol $\tilde{\Gamma}$ let us look at the former term and investigate its symmetries. We rewrite it in the following form,
	\begin{align*}
	- D_{{\alpha}}g_{\mu\nu}+g_{\beta\nu}\tilde{\Gamma}^\beta_{\alpha\mu}+g_{\mu\beta}\tilde{\Gamma}^\beta_{\alpha\nu}&=-g_{\lambda\beta}\Sigma^{\lambda\beta}_{\alpha\nu\mu}\\
	& +g_{\lambda\beta}\Gamma^\beta_{\rho\nu}([x^\rho,\Gamma^\lambda_{\alpha\mu}]+D^{\lambda\rho}_{\ \ \sigma}\Gamma^\sigma_{\alpha\mu}-D^{\sigma\rho}_{\ \ \alpha}\Gamma^\lambda_{\sigma\mu}-D^{\sigma\rho}_{\ \ \mu}\Gamma^\lambda_{\alpha\sigma}).
	\end{align*}
	Since the left-hand side is symmetric w.r.t.\ the indices $\mu$ and $\nu$ we conclude that the $\Sigma$-term has to be given by 
	\begin{align*}
	\Sigma^{\lambda\beta}_{\alpha\nu\mu}&=-\Gamma^\beta_{\rho\mu}([x^\rho,\Gamma^\lambda_{\alpha\nu}]+D^{\lambda\rho}_{\ \ \sigma}\Gamma^\sigma_{\alpha\nu}-D^{\sigma\rho}_{\ \ \alpha}\Gamma^\lambda_{\sigma\nu}-D^{\sigma\rho}_{\ \ \nu}\Gamma^\lambda_{\alpha\sigma}) .
	\end{align*} 
	Hence, metricity finally reads, 
	\begin{align*}
	0&=D_{ {\alpha}}g_{\mu\nu}-g_{\beta\nu}\tilde{\Gamma}^\beta_{\alpha\mu}-g_{\mu\beta}\tilde{\Gamma}^\beta_{\alpha\nu}-g_{\lambda\beta}\Sigma^{\lambda\beta}_{\alpha\nu\mu}-g_{\lambda\beta}\Sigma^{\lambda\beta}_{\alpha\mu\nu}.
	\end{align*}
	By using the result for metricity $(\nabla g =0)$ and cycling indices we have
	\begin{align*}
	0&=D_{ {\alpha}}g_{\mu\nu}-g_{\beta\nu}\tilde{\Gamma}^\beta_{\alpha\mu}-g_{\mu\beta}\tilde{\Gamma}^\beta_{\alpha\nu}-g_{\lambda\beta}\Sigma^{\lambda\beta}_{\alpha\nu\mu}-g_{\lambda\beta}\Sigma^{\lambda\beta}_{\alpha\mu\nu} \\
	&+D_{ {\mu}}g_{\nu\alpha}-g_{\beta\alpha}\tilde{\Gamma}^\beta_{\mu\nu}-g_{\nu\beta}\tilde{\Gamma}^\beta_{\mu\alpha}-g_{\lambda\beta}\Sigma^{\lambda\beta}_{\mu\alpha\nu}-g_{\lambda\beta}\Sigma^{\lambda\beta}_{\mu\nu\alpha} \\
	&-D_{ {\nu}}g_{\alpha\mu}+g_{\beta\mu}\tilde{\Gamma}^\beta_{\nu\alpha}+g_{\alpha\beta}\tilde{\Gamma}^\beta_{\nu\mu}+g_{\lambda\beta}\Sigma^{\lambda\beta}_{\nu\mu\alpha}+g_{\lambda\beta}\Sigma^{\lambda\beta}_{\nu\alpha\mu} ,
	\end{align*}
	which is in direct analogy to the commutative case. By re-ordering the terms we obtain, 
	\begin{align*}
	2g_{\beta\nu} \tilde{\Gamma}^\beta_{\alpha\mu} & 
	=D_{ {\alpha}}g_{\mu\nu }+D_{ {\mu}}g_{\nu\alpha}-D_{  {\nu}}g_{\alpha\mu}-2g_{\lambda\beta}\Sigma^{\lambda\beta}_{\alpha\mu\nu} .
	\end{align*}
	By multiplying  from the left with the inverse metric, this renders 
	\begin{align*}
	\tilde{\Gamma}^{\kappa}_{\alpha\mu} & 
	=\frac{1}{2}g^{\kappa\nu}\left(D_{ {\alpha}}g_{\mu\nu }+D_{ {\mu}}g_{\nu\alpha}-D_{ {\nu}}g_{\alpha\mu}-2g_{\lambda\beta}\Sigma^{\lambda\beta}_{\alpha\mu\nu} \right)\\& 
	=: \overline {\Gamma}_{\alpha\mu}^{\kappa}-g^{\kappa\nu}  \Sigma_{\alpha\mu\nu},
	\end{align*}
	with $ \Sigma_{\alpha\mu\nu}\equiv g_{\lambda\beta}\Sigma^{\lambda\beta}_{\alpha\mu\nu} $. 		
\end{proof} 
\begin{remark}
	In Theorem (\ref{P3.1}) we solved for the $\Sigma$-term by making use of the symmetric structure of the indices on the l.h.s. The reader may object that some other terms besides the counter-terms on the r.h.s. should be taken into account as long as they are symmetric in the necessary indices. However, recalling that the definition for this quantity is $[dx^\lambda,\Gamma^\beta_{\alpha\nu }]=\Sigma^{\lambda\beta}_{\alpha\nu\mu}dx^\mu$, explicitly implies that it is symmetric in $\alpha,\,\nu$. From this we see that up to first order, there is no other possible terms. Moreover, note that $\overline{\Gamma}$ is of first order in the noncommutative parameters due to the fact that the derivations yield terms to this order. However, in the following we write $\Gamma$ instead of $\overline{\Gamma}$ if the term is multiplied by the structure constants $C$ or $S$.
\end{remark}
The following result is needed in order to prove Proposition \ref{prop:cdww}.
\begin{proposition} \label{prop:ccd}
	The commutator of the generators of the algebra $\mathfrak{A}$ and the connection symbol $\overline{\Gamma}$ introduced in Theorem \ref{P3.1}, is given by  
	\begin{align} 
	[x^{\rho}, \overline{\Gamma}_{\alpha\nu}^{\lambda} ]=& -	D^{\lambda\rho}_{\ \ \sigma} {\Gamma}^{\sigma}_{\alpha\nu} 	+D^{\sigma\rho}_{\ \ \alpha} {\Gamma}^{\lambda}_{\sigma\nu}+D^{\sigma\rho}_{\ \ \nu} {\Gamma}^{\lambda}_{\sigma\alpha}\nonumber\\&\qquad\qquad 
	- g^{\lambda\kappa}
	\biggl( S^{\sigma\rho}_{\alpha}D_{\sigma}g_{\kappa\nu}   
	+S^{\sigma\rho}_{\nu}D_{\sigma}g_{\kappa\alpha} -  	S^{\sigma\rho}_{\kappa}
	D_{\sigma}g_{\alpha\nu} 
	\biggr).\label{eq:ccd_eq1}
	\end{align} 
\end{proposition}
	\begin{proof}
	See Appendix \ref{proof:cdww1}.
\end{proof}	
\begin{proposition}\label{prop:cdww}
	Let $M\in\Omega^1(\mathcal{A})\otimes \Omega^1(\mathcal{A})$, the action of the covariant derivative on it, up to first order, is given by
	\begin{align*}
	\nabla_\rho M_{\mu\nu}&=D_{\rho}(M_{\mu\nu})-M_{\lambda\nu}\tilde{\Gamma}^\lambda_{\rho\mu}-M_{\mu\lambda}\tilde{\Gamma}^\lambda_{\rho\nu}
	-M_{\lambda\beta}	\Sigma^{\lambda\beta}_{\mu\rho\nu} - M_{\lambda\beta}	\Sigma^{\lambda\beta}_{\nu\rho\mu}+\mathcal{O}(D^2).
	\end{align*} For $T\in\Omega^1(\mathcal{A})\otimes \Omega^1(\mathcal{A})\otimes \Omega^1(\mathcal{A})$, the action of the covariant derivative up to first order  is given by
	\begin{align*}
	\nabla_{\lambda}(T_{\alpha\sigma\kappa})&=  D_{\lambda}T_{\alpha\sigma\kappa}  
	-
	T_{\rho\sigma\kappa}\tilde{\Gamma}^{\rho}_{\lambda\alpha}
	-T_{\alpha\rho\kappa} \tilde{\Gamma}^{\rho}_{\lambda\sigma}  -T_{\alpha\sigma\rho}\tilde{\Gamma}^{\rho}_{\lambda\kappa}   -T_{\rho\nu\kappa}\Sigma^{\rho\nu}_{\lambda\sigma\alpha} 
	-T_{\rho\sigma\nu}\Sigma^{\rho\nu}_{\lambda\kappa\alpha}\\&-T_{\alpha\rho\nu} \Sigma^{\rho\nu}_{\lambda\kappa\sigma}
	-
	T_{\rho\nu\kappa} \Sigma^{\rho\nu}_{\lambda\alpha\sigma }   -T_{\rho\sigma\nu} \Sigma^{\rho\nu}_{\lambda\alpha\kappa}      
	-T_{\alpha\rho\nu} \Sigma^{\rho\nu}_{\lambda\sigma\kappa} .
	\end{align*}
	\begin{proof}
		See Appendix  \ref{proof:cdww}.
	\end{proof}	
\end{proposition} 
The general formula is the formula for the covariant derivative of the most general Lie-algebraic type of a noncommutative spacetime. If one sets the deformation constants, i.e. the structure constants $D$, equal to zero one obtains the classical case.  
	\section{Riemann Tensor}\label{sec:4}
In this section, we define the Riemann tensor in the context of NCG. It is given as a combination of the exterior derivative $d$, the wedge product $\wedge$, which maps the tensor product of two elements of the algebra of forms to the skew-symmetric product,  and    the covariant derivative. To give a definition of the noncommutative curvature we use the extension of the covariant derivative, see Definition \ref{ext}. 
\begin{defi}\label{curv0}
	Let $dx^\mu\in\Omega^1(\mathfrak{A})$ and $\nabla$ be the connection, following \cite[Proposition 42]{landi_2014} and \cite[Section 4.1]{Mo} we define the \textbf{curvature} as
	\begin{align*}
	R(dx^\mu):=D^2(dx^\mu)=D\circ D(dx^\mu).
	\end{align*}
\end{defi}
Note that in \cite[Section 3.7]{madore_1999} it is mentioned that there is not an appropriate way to define the curvature, since there is no guarantee that taking the commutative definition will provide a meaningful quantity in a noncommutative setting. However, we use the formerly mentioned definitions in order to obtain a concrete formula for the Riemann tensor that allows us to  calculate  noncommutative corrections to various physical objects.   
\begin{proposition} \label{prop:riem} 	Let $dx^\mu\in \Omega^1(\mathfrak{A})$ and $\nabla$ be the connection, then  the \textbf{Riemann tensor} $R:\Omega^1(\mathfrak{A})\to \Omega^2(\mathfrak{A})\otimes \Omega^1(\mathfrak{A})$,   given by Definition \ref{curv0}, is equivalent to the one   in \cite{BM}, i.e.\
	\begin{align} 
	&R(dx^\mu)=(d\otimes\mathbb{I}-(\wedge\otimes\mathbb{I})\circ(\mathbb{I}\otimes\nabla))\circ\nabla(dx^\mu). \label{curv}
	\end{align}
	
	\begin{proof}
		According to Definition \ref{ext} the mapping $D:\Omega^{*}(\mathfrak{A})\otimes_{\mathfrak{A}}\Omega^1(\mathfrak{A})$,  which is  an extension of the notion of a covariant derivative (where on a one form it is by definition equal to the covariant derivative), acts as follows on a two-form 
		\begin{align*}
		D^2(dx^\mu)=D\circ D(dx^\mu)=D\circ\nabla(dx^\mu)&=D(\omega\otimes \omega')\\
		&=d(\omega)\otimes \omega'-\omega\cdot\nabla(\omega')\\
		&=\left(d\otimes \mathbb{I}-(\wedge\otimes\mathbb{I})\circ(\mathbb{I}\otimes\nabla)\right)(\omega\otimes \omega')\\
		&=\left(d\otimes \mathbb{I}-(\wedge\otimes\mathbb{I})\circ(\mathbb{I}\otimes\nabla)\right)\circ\nabla(dx^\mu),
		\end{align*}
		where we used the notation $\nabla(dx^\mu)=\omega\otimes \omega'$.
		
	\end{proof}
\end{proposition}
By using the former definition and the explicit formula for the covariant derivative for a general Lie-algebraic noncommutative spacetime we calculate the Riemann tensor explicitly.  
\begin{proposition} \label{rt} 
	The \textbf{Riemann tensor components}   for the most general Lie-algebraic type of noncommutative spacetime are  explicitly given by
	\begin{equation}
	\tilde{R}^{\mu}_{\;\;\sigma\alpha\rho}=D_{\alpha}\tilde{\Gamma}^{\mu}_{\rho\sigma  }-D_{  \rho}\tilde{\Gamma}^{\mu}_{\alpha\sigma}+\tilde{\Gamma}^{\mu}_{\alpha\lambda}\tilde{\Gamma}^\lambda_{\rho\sigma}-\tilde{\Gamma}^{\mu}_{\rho\lambda}\tilde{\Gamma}^\lambda_{\alpha\sigma} \label{riemann1}.
	\end{equation}   
	\begin{proof}  
		The action of the curvature upon the coordinated basis of one-forms is given in terms of its defining Equation (\ref{curv}) 
		\begin{align*}
		R(dx^\mu)=&(d\otimes\mathbb{I}-(\wedge\otimes\mathbb{I})\circ(\mathbb{I}\otimes\nabla))\nabla(dx^\mu),
		\end{align*}
		by inserting the equation $\nabla(dx^\mu)=-\tilde{\Gamma}^{\mu}_{\rho\sigma}dx^\rho\otimes dx^\sigma$ into    the definition of the Riemann tensor one has
		\begin{align*}
		R(dx^\mu)=&-(d\otimes\mathbb{I}-(\wedge\otimes\mathbb{I})\circ(\mathbb{I}\otimes\nabla))(\tilde{\Gamma}^{\mu}_{\rho\sigma}dx^\rho\otimes dx^\sigma)\\
		=&-d(\tilde{\Gamma}^{\mu}_{\rho\sigma}dx^\rho)\otimes dx^\sigma+(\tilde{\Gamma}^{\mu}_{\rho\sigma}dx^\rho\wedge\nabla(dx^\sigma))\\
		=&-d(\tilde{\Gamma}^{\mu}_{\rho\sigma})\wedge dx^\rho\otimes dx^\sigma-\tilde{\Gamma}^{\mu}_{\rho\sigma}dx^\rho\wedge\tilde{\Gamma}^\sigma_{\alpha\beta}dx^\alpha \otimes dx^\beta.
		\end{align*}
		Since we obtain  an entity that lies in $\Omega^2({\mathfrak A} )\otimes\Omega^1({\mathfrak A} )$, we  move all elements of the algebra to the left, c.f. \cite[Chapter 7, Equation (7.12)]{landi_2014}. In order to proceed we  commute the Christoffel symbol with an element of the basis of one-forms.
		Given that our calculation is up to first order, we  only consider the classical part of the Christoffel symbol multiplying the commutator, 
		\begin{align*}
		R(dx^\mu)=&-d\tilde{\Gamma}^{\mu}_{\rho\sigma}\wedge dx^\rho\otimes dx^\sigma-\tilde{\Gamma}^{\mu}_{\rho\sigma}\tilde{\Gamma}^\sigma_{\alpha\beta}dx^\rho\wedge dx^\alpha \otimes dx^\beta .
		\end{align*}
		By rearranging the indices   we obtain
		\begin{align*}
		R(dx^\mu)=&-(D_{\lambda}\tilde{\Gamma}^{\mu}_{\alpha\beta}+\tilde{\Gamma}^{\mu}_{\lambda\sigma}\tilde{\Gamma}^\sigma_{\alpha\beta})dx^\lambda\wedge dx^\alpha \otimes dx^\beta \\
		=&-(D_{\lambda}\tilde{\Gamma}^{\mu}_{\alpha\beta }+\tilde{\Gamma}^{\mu}_{\lambda\sigma}\tilde{\Gamma}^\sigma_{\alpha\beta} -D_{\alpha}\tilde{\Gamma}^{\mu}_{\lambda\beta }-\tilde{\Gamma}^{\mu}_{\alpha\sigma}\tilde{\Gamma}^\sigma_{\lambda\beta})dx^\lambda\otimes dx^\alpha \otimes dx^\beta.
		\end{align*}
	\end{proof}
\end{proposition}  
Next, we   prove   identities and symmetries  of the noncommutative Riemannn tensor that are essential for the covariance of the noncommutative Einstein equations. 
\begin{theorem} \label{prop:bianchi1}	The \textbf{first Bianchi identity}  reads
	\begin{align*}
	\tilde{R}^{\mu}_{\;\;[\sigma\alpha\rho]}
	&=0.
	\end{align*}
\end{theorem}
\begin{proof}
	The proof is straightforward  
	\begin{align*}
	\tilde{R}^{\mu}_{\;\;[\sigma\alpha\rho]}&=\frac{1}{3}\left(\tilde{R}^{\mu}_{\;\;\sigma\alpha\rho}+\tilde{R}^{\mu}_{\;\;\alpha\rho\sigma}+\tilde{R}^{\mu}_{\;\;\rho\sigma\alpha}\right)\\
	&=\frac{1}{3}\left(D_{\alpha}\tilde{\Gamma}^{\mu}_{\rho\sigma}-D_{\rho}\tilde{\Gamma}^{\mu}_{\alpha\sigma}+\tilde{\Gamma}^{\mu}_{\alpha\nu}\tilde{\Gamma}^\nu_{\rho\sigma}-\tilde{\Gamma}^{\mu}_{\rho\nu}\tilde{\Gamma}^\nu_{\alpha\sigma}  \right)
	\\
	&+\frac{1}{3}\left(D_{\rho}\tilde{\Gamma}^{\mu}_{\sigma\alpha}-D_{\sigma}\tilde{\Gamma}^{\mu}_{\rho\alpha}+\tilde{\Gamma}^{\mu}_{\rho\nu}\tilde{\Gamma}^\nu_{\sigma\alpha}-\tilde{\Gamma}^{\mu}_{\sigma\nu}\tilde{\Gamma}^\nu_{\rho\alpha}  \right) 
	\\
	&+\frac{1}{3}\left(D_{\sigma}\tilde{\Gamma}^{\mu}_{\alpha\rho}-D_{\alpha}\tilde{\Gamma}^{\mu}_{\sigma\rho}+\tilde{\Gamma}^{\mu}_{\sigma\nu}\tilde{\Gamma}^\nu_{\alpha\rho}-\tilde{\Gamma}^{\mu}_{\alpha\nu}\tilde{\Gamma}^\nu_{\sigma\rho}   \right)\\
	&=0.
	\end{align*}

\end{proof}
To prove the second Bianchi identity we write the Riemann tensor in a convenient form and prove the following intermediate result. 
\begin{proposition}\label{prop:comm1}
	The commutator of covariant derivatives acting on arbitrary element of the algebra is
	\begin{align*}
	[\nabla_\mu,\nabla_\nu]A_\lambda
	&=A_\tau \tilde{R}^\tau_{\lambda\nu\mu}+(\nabla_\tau A_\kappa)U^{\tau\kappa}_{\lambda\nu\mu}	,
	\end{align*}
	where 	$U^{\tau\kappa}_{\lambda\nu\mu}:=(\Sigma^{\tau\kappa}_{\lambda\nu\mu}-\Sigma^{\tau\kappa}_{\lambda\mu\nu})$, is skew-symmetric in $\mu$ and $\nu$.
	\begin{proof}
		Let us analyse the action of the commutator of the covariant derivatives on an arbitrary element of the algebra
		\begin{align*}
		[\nabla_\mu,\nabla_\nu]A_\lambda&=\nabla_\mu(\nabla_\nu A_\lambda)-(\mu \leftrightarrow \nu)\\
		&=D_{\mu}(K_{\nu\lambda})-K_{\tau\lambda}\tilde{\Gamma}^\tau_{\mu\nu}-K_{\nu\tau}\tilde{\Gamma}^\tau_{\mu\lambda}-K_{\eta\gamma}\Sigma^{\eta\gamma}_{\mu\nu\lambda}\\
		&- K_{\tau\kappa}\Sigma^{\tau\kappa}_{\lambda\mu\nu} -(\mu \leftrightarrow \nu)\\
		&=D_{\mu}(K_{\nu\lambda})-K_{\nu\tau}\tilde{\Gamma}^\tau_{\mu\lambda}
		- K_{\tau\kappa}\Sigma^{\tau\kappa}_{\lambda\mu\nu}-(\mu \leftrightarrow \nu)
		\end{align*}
		where $K_{\nu\lambda}:=\nabla_\nu A_\lambda$ and we have eliminated all the terms that are symmetric in $\mu$ and $\nu$. The first term is 
		\begin{align*}
		D_{\mu}(K_{\nu\lambda})=D_{\mu}(D_\nu A_\lambda-A_\tau\tilde{\Gamma}^\tau_{\nu\lambda})=D_{\mu}D_\nu A_\lambda-(D_{\mu}A_\tau)\tilde{\Gamma}^\tau_{\nu\lambda}   -A_\tau (D_{\mu}\tilde{\Gamma}^\tau_{\nu\lambda})
		\end{align*}
		we drop the terms symmetric in $\mu$ and $\nu$ to get
		\begin{align*}
		[\nabla_\mu,\nabla_\nu]A_\lambda
		&=-(D_{\mu}A_\tau)\tilde{\Gamma}^\tau_{\nu\lambda}   -A_\tau (D_{\mu}\tilde{\Gamma}^\tau_{\nu\lambda})-(\nabla_\nu A_\tau)\tilde{\Gamma}^\tau_{\mu\lambda} - K_{\tau\kappa}\Sigma^{\tau\kappa}_{\lambda\mu\nu}\\
		&-(\mu \leftrightarrow \nu)\\
		&=(D_{\nu}A_\tau)\tilde{\Gamma}^\tau_{\mu\lambda}-(D_{\mu}A_\tau)\tilde{\Gamma}^\tau_{\nu\lambda}+A_\tau (D_{\nu}\tilde{\Gamma}^\tau_{\mu\lambda})   -A_\tau (D_{\mu}\tilde{\Gamma}^\tau_{\nu\lambda})\\
		&+(\nabla_\mu A_\tau)\tilde{\Gamma}^\tau_{\nu\lambda}-(\nabla_\nu A_\tau)\tilde{\Gamma}^\tau_{\mu\lambda}+K_{\tau\kappa}\Sigma^{\tau\kappa}_{\lambda\nu\mu}-K_{\tau\kappa}\Sigma^{\tau\kappa}_{\lambda\mu\nu},
		\end{align*}
		where so far we have just written explicitly all the terms involved in the calculation. Next, we expand the covariant derivatives on $A$ which contain terms that   cancel the first two terms and   lead us to a quantity that is the Riemann tensor and an additional term, i.e.\
		\begin{align*}
		[\nabla_\mu,\nabla_\nu]A_\lambda&=A_\tau (D_{\nu}\tilde{\Gamma}^\tau_{\mu\lambda})   -A_\tau (D_{\mu}\tilde{\Gamma}^\tau_{\nu\lambda})-A_\tau\tilde{\Gamma}^\tau_{\mu\gamma}\tilde{\Gamma}^\gamma_{\nu\lambda}+A_\tau\tilde{\Gamma}^\tau_{\nu\gamma}\tilde{\Gamma}^\gamma_{\mu\lambda}\\
		&+ (D_\tau A_\kappa-A_\theta\tilde{\Gamma}^\theta_{\tau\kappa}) (\Sigma^{\tau\kappa}_{\lambda\nu\mu}-\Sigma^{\tau\kappa}_{\lambda\mu\nu})  \\
		&:=A_\tau \tilde{R}^\tau_{\lambda\nu\mu}+(\nabla_\tau A_\kappa)U^{\tau\kappa}_{\lambda\nu\mu}.		
		\end{align*}
	\end{proof}
\end{proposition}	

\begin{theorem}\label{prop:secbi}
	The \textbf{second Bianchi identity} is given by 
	\begin{align*}	\tilde{R}^\tau_{\lambda[\nu\mu;\beta]}=0,  \qquad U^{\tau\gamma}_{\lambda[\nu\mu;\beta]}=0.
	\end{align*}
	
	\begin{proof}
		We know from  the Jacobi identities that  cyclic permutation of  covariant derivatives vanish, i.e.\ 
		\begin{align*}
		0=[\nabla_{[\beta},[\nabla_\mu,\nabla_{\nu]}]].
		\end{align*}
		The former expression applied on an element of the algebra gives us 
		\begin{align*}
		0&=	[\nabla_{\beta},[\nabla_{\mu},\nabla_{\nu}]]A_{\lambda}
		+	[\nabla_{\mu},[\nabla_{\nu},\nabla_{\beta}]]A_{\lambda}+
		[\nabla_{\nu},[\nabla_{\beta},\nabla_{\mu}]]A_{\lambda}
		\\&= \nabla_{\beta}\left([\nabla_{\mu},\nabla_{\nu}]A_{\lambda}\right)
		+\nabla_{\beta}\left([\nabla_{\nu},\nabla_{\mu}]A_{\lambda}\right)
		\\&+\nabla_{\mu}\left([\nabla_{\nu},\nabla_{\beta}]A_{\lambda}\right)
		+\nabla_{\mu}\left([\nabla_{\beta},\nabla_{\nu}]A_{\lambda}\right)\\&+\nabla_{\nu}\left([\nabla_{\mu},\nabla_{\beta}]A_{\lambda}\right)
		+\nabla_{\nu}\left([\nabla_{\beta},\nabla_{\mu}]A_{\lambda}\right),
		\end{align*}
		where next we use Proposition (\ref{prop:comm1}) to write down the commutators
		\begin{align*}
		0&= 
		\nabla_{\beta}\left(A_\tau \tilde{R}^\tau_{\lambda\nu\mu}+  U^{\tau\gamma}_{\lambda\nu\mu} \nabla_\tau A_\gamma \right)
		+\nabla_{\beta}\left(A_\tau \tilde{R}^\tau_{\lambda\mu\nu}+  U^{\tau\gamma}_{\lambda\mu\nu} \nabla_\tau A_\gamma  \right)
		\\&+\nabla_{\mu}\left(A_\tau \tilde{R}^\tau_{\lambda\beta\nu}+  U^{\tau\gamma}_{\lambda\beta\nu} \nabla_\tau A_\gamma  \right)
		+\nabla_{\mu}\left(A_\tau \tilde{R}^\tau_{\lambda\nu\beta}+  U^{\tau\gamma}_{\lambda\nu\beta} \nabla_\tau A_\gamma \right)\\&+\nabla_{\nu}\left(A_\tau \tilde{R}^\tau_{\lambda\beta\mu}+  U^{\tau\gamma}_{\lambda\beta\mu} \nabla_\tau A_\gamma \right)
		+\nabla_{\nu}\left(A_\tau \tilde{R}^\tau_{\lambda\mu\beta}+  U^{\tau\gamma}_{\lambda\mu\beta} \nabla_\tau A_\gamma  \right).
		\end{align*}
		By  regrouping   terms that appear due to the covariant derivatives acting on $\tilde{R}$ and by reordering  those terms where the derivatives are on $A$, we have
		\begin{align*}
		0&= 
		A_\tau\tilde{R}^\tau_{\lambda[\nu\mu;\beta]}
		+
		\nabla_{\beta}A_\tau\left( \tilde{R}^\tau_{\lambda\nu\mu} +\tilde{R}^\tau_{\lambda\mu\nu} \right) 
		+
		\nabla_{\mu}A_\tau\left( \tilde{R}^\tau_{\lambda\beta\nu} +\tilde{R}^\tau_{\lambda\nu\beta} \right) 
		\\&+
		\nabla_{\nu}A_\tau\left( \tilde{R}^\tau_{\lambda\beta\mu} +\tilde{R}^\tau_{\lambda\mu\beta} \right) +
		\left(U^{\tau\gamma}_{\lambda\nu\mu}+U^{\tau\gamma}_{\lambda\mu\nu} \right)	\nabla_{\beta}\nabla_\tau A_\gamma+
		\left(U^{\tau\gamma}_{\lambda\beta\nu}+U^{\tau\gamma}_{\lambda\nu\beta } \right)	\nabla_{\mu}\nabla_\tau A_\gamma\\&+
		\left(U^{\tau\gamma}_{\lambda\beta\mu}+U^{\tau\gamma}_{\lambda\mu\beta} \right)	\nabla_{\nu}\nabla_\tau A_\gamma
		+
		U^{\tau\gamma}_{\lambda[\nu\mu;\beta]}
		\nabla_\tau A_\gamma
		\\&= 	A_\tau\tilde{R}^\tau_{\lambda[\nu\mu;\beta]}+U^{\tau\gamma}_{\lambda[\nu\mu;\beta]}
		\nabla_\tau A_\gamma   ,
		\end{align*}	
		where in the last lines we used the properties of the   (two) tensors $R$ and $U$. Since the curvature term does not contain any derivatives the two terms on the right-hand side of the equation have to be  separately zero.  
	\end{proof}
\end{theorem}
The next result is important in order to ease the upcoming proofs and calculations. 
\begin{proposition} \label{prop:44}
	The general noncommutative Riemann tensor (c.f. Proposition \ref{prop:riem})  rewritten in terms of barred Riemann tensor quantities and $\Sigma$ terms reads as follows
	\begin{align*}
	\tilde{R}^{\mu}_{\;\;\beta\alpha\lambda}&=\overline{R}^{\mu}_{\;\;\beta\alpha\lambda}+g^{\mu\nu}T_{\beta\alpha\lambda\nu},
	\end{align*}
	where we introduced $T_{\beta\alpha\lambda\nu}:=\nabla_{\lambda}  \Sigma_{\alpha\beta\nu} -\nabla_{\alpha}  \Sigma_{\lambda\beta\nu}$.
	\begin{proof}
		We calculate directly by substituting the barred connection symbol from Theorem \ref{P3.1} into Proposition \ref{prop:riem},
		\begin{align*}
		\tilde{R}^{\mu}_{\;\;\beta\alpha\lambda}&=D_{\alpha} \tilde{\Gamma}^{\mu}_{\lambda\beta}   +\tilde{\Gamma}^{\mu}_{\alpha\sigma}\tilde{\Gamma}^\sigma_{\lambda\beta}- (\alpha\leftrightarrow\lambda)\\
		&=D_{\alpha}(\overline{\Gamma}^{\mu}_{\lambda\beta}-g^{\mu\nu}  \Sigma_{\lambda\beta\nu}) 
		+(\overline{\Gamma}^{\mu}_{\alpha\sigma}-g^{\mu\nu}  \Sigma_{\alpha\sigma\nu})(\overline{\Gamma}^\sigma_{\lambda\beta}-g^{\sigma\kappa}  \Sigma_{\lambda\beta\kappa})
		- (\alpha\leftrightarrow\lambda),
		\end{align*}
		where we have made use of the definitions contained within Theorem \ref{P3.1} to write $\tilde{\Gamma}$ in terms of $\overline{\Gamma}$, 
		\begin{align*}
		\tilde{R}^{\mu}_{\;\;\beta\alpha\lambda}&=D_{\alpha} \overline{\Gamma}^{\mu}_{\lambda\beta}-D_{\alpha} (g^{\mu\nu}  \Sigma_{\lambda\beta\nu})  
		+\overline{\Gamma}^{\mu}_{\alpha\sigma}\overline{\Gamma}^\sigma_{\lambda\beta}-g^{\mu\nu} \overline{\Gamma}^\sigma_{\lambda\beta} \Sigma_{\alpha\sigma\nu}
		-g^{\sigma\nu}  \overline{\Gamma}^{\mu}_{\alpha\sigma}\Sigma_{\lambda\beta\nu} 
		- (\alpha\leftrightarrow\lambda),
		\end{align*}
		by reordering some terms and using the $\overline{R}$-term we have,
		\begin{align*}
		&=	\overline{R}^{\mu}_{\;\;\beta\alpha\lambda}-\biggl(g^{\mu\nu}D_{\alpha}  \Sigma_{\lambda\beta\nu}  
		- g^{\mu\nu}  {\Gamma}^\sigma_{\alpha\beta} \Sigma_{\lambda\sigma\nu}-  g^{\mu\nu}  {\Gamma}^\sigma_{\alpha\nu} \Sigma_{\lambda\beta\sigma} - (\alpha\leftrightarrow\lambda) \biggr) \\&=	\overline{R}^{\mu}_{\;\;\beta\alpha\lambda}-g^{\mu\nu}\biggl(\nabla_{\alpha}  \Sigma_{\lambda\beta\nu}  -\nabla_{\lambda}  \Sigma_{\alpha\beta\nu}   \biggr) \\&=	\overline{R}^{\mu}_{\;\;\beta\alpha\lambda}+g^{\mu\nu}\biggl( \nabla_{\lambda}  \Sigma_{\alpha\beta\nu} -\nabla_{\alpha}  \Sigma_{\lambda\beta\nu}   \biggr)\\
		\tilde{R}^{\mu}_{\;\;\beta\alpha\lambda}&=:	\overline{R}^{\mu}_{\;\;\beta\alpha\lambda}+g^{\mu\nu}T_{\beta\alpha\lambda\nu}
		\end{align*}
	\end{proof}
\end{proposition} 	 
Next, we use the  noncommutative version of the   Bianchi identities   and   the  intermediate results to continue our investigations of  symmetries of the noncommutative curvature tensor. This is important for the upcoming proof  of  the covariance of the Einstein tensor. We follow here the proof of \cite{Wald:1984rg}, i.e. we calculate the commutator of two covariant derivatives on the metric. 
\begin{proposition} \label{prop:riemsym}
	The noncommutative Riemann tensor $\tilde{R}$ satisfies	the following symmetries
	\begin{align*}	
	\tilde{R}_{\sigma\kappa\alpha\lambda}		=-\tilde{R}_{\kappa\sigma\alpha\lambda}  	-K_{\kappa\sigma\alpha\lambda},  
	\end{align*}
	where we defined 
	\begin{align*}	
	K_{\kappa\sigma\alpha\lambda}:&=   \nabla_{\alpha} \Sigma_{\lambda\kappa\sigma} +  \nabla_{\alpha} \Sigma_{\lambda\sigma\kappa}
	-\nabla_{\lambda} \Sigma_{\alpha\kappa\sigma} 
	- \nabla_{\lambda} \Sigma_{\alpha\sigma\kappa}.
	\end{align*} 
	Moreover, the barred Riemann tensor $\overline{R}$  satisfies the following symmetry 
	\begin{align*}	
	\overline{R}_{\kappa\sigma\alpha\lambda}=-
	\overline{R}_{\sigma\kappa\alpha\lambda}.
	\end{align*} 	 
	\begin{proof}
		See Appendix \ref{proof:riemsym}.
	\end{proof}
\end{proposition}

\begin{proposition} 
	The barred Riemann tensor obeys the following  symmetry
	$$\overline{R}_{\mu\alpha\lambda\beta} =\overline{R}_{\lambda\beta\mu\alpha}.$$
	\begin{proof}
		We  study the difference between the following Riemann tensors
		\begin{align*}
		\overline{R}_{\mu\beta\alpha\lambda}-\overline{R}_{\alpha\lambda\mu\beta}&=g_{\mu\sigma}\overline{R}^{\sigma}_{\;\;\beta\alpha\lambda}-g_{\alpha\sigma}\overline{R}^{\sigma}_{\;\;\lambda\mu\beta}=g_{\mu\sigma}\overline{R}^{\sigma}_{\;\;\beta\alpha\lambda}-g_{\alpha\sigma}(-\overline{R}^{\sigma}_{\;\;\beta\lambda\mu}-\overline{R}^{\sigma}_{\;\;\mu\beta\lambda})\\
		&=g_{\mu\sigma}\overline{R}^{\sigma}_{\;\;\beta\alpha\lambda}+g_{\alpha\sigma}\overline{R}^{\sigma}_{\;\;\beta\lambda\mu}+g_{\alpha\sigma}\overline{R}^{\sigma}_{\;\;\mu\beta\lambda}=\overline{R}_{\mu\beta\alpha\lambda}+\overline{R}_{\alpha\beta\lambda\mu}+\overline{R}_{\alpha\mu\beta\lambda},
		\end{align*}
		where in the last lines we used the \textit{first Bianchi identity} from Proposition \ref{prop:bianchi1}. Moreover, in Proposition \ref{prop:riemsym} we  proved that the Riemann tensor is skew-symmetric in the  first two indices, a fact we   make use of for the last two terms to obtain
		\begin{align*}
		\overline{R}_{\mu\beta\alpha\lambda}-\overline{R}_{\alpha\lambda\mu\beta}&=\overline{R}_{\mu\beta\alpha\lambda}-\overline{R}_{\beta\alpha\lambda\mu}-\overline{R}_{\mu\alpha\beta\lambda},
		\end{align*}
		once more, we make use of the \textit{first Bianchi indentity} for the second and the third term to obtain
		\begin{align*}
		\overline{R}_{\mu\beta\alpha\lambda}-\overline{R}_{\alpha\lambda\mu\beta}&=\overline{R}_{\mu\beta\alpha\lambda}-(-\overline{R}_{\beta\mu\alpha\lambda}-\overline{R}_{\beta\lambda\mu\alpha})-(-\overline{R}_{\mu\lambda\alpha\beta}-\overline{R}_{\mu\beta\lambda\alpha})\\
		&=\overline{R}_{\mu\beta\alpha\lambda}+\overline{R}_{\beta\mu\alpha\lambda}+\overline{R}_{\beta\lambda\mu\alpha}+\overline{R}_{\mu\lambda\alpha\beta}+\overline{R}_{\mu\beta\lambda\alpha}\\
		&=\overline{R}_{\beta\lambda\mu\alpha}+\overline{R}_{\mu\lambda\alpha\beta}-\overline{R}_{\mu\beta\alpha\lambda},
		\end{align*}
		where in the transition from the second to the third line the first two terms cancelled each other and the fifth acquires the sign under exchange of its last two indices due to the fact that the Riemann tensor is also skew-symmetric on those too, by its very definition. Using skew-symmetry in the last two indices for the second term and then the \textit{first Bianchi identity} along with the third yields
		\begin{align*}
		\overline{R}_{\mu\beta\alpha\lambda}-\overline{R}_{\alpha\lambda\mu\beta}&=\overline{R}_{\beta\lambda\mu\alpha}-\overline{R}_{\mu\lambda\beta\alpha}-\overline{R}_{\mu\beta\alpha\lambda}\\
		&=\overline{R}_{\beta\lambda\mu\alpha}+\overline{R}_{\mu\alpha\lambda\beta}\\
		&=-\overline{R}_{\beta\lambda\alpha\mu}+\overline{R}_{\mu\alpha\lambda\beta}.
		\end{align*}
		Using this last result but interchanging $\mu\leftrightarrow\beta$ and $\alpha \leftrightarrow\lambda$ leads to
		\begin{align*}
		\overline{R}_{\beta\mu\lambda\alpha}-\overline{R}_{\lambda\alpha\beta\mu}
		&=-\overline{R}_{\mu\alpha\lambda\beta}+\overline{R}_{\beta\lambda\alpha\mu}\\
		\overline{R}_{\mu\beta\alpha\lambda}-\overline{R}_{\alpha\lambda\mu\beta}
		&=-\overline{R}_{\mu\alpha\lambda\beta}+\overline{R}_{\beta\lambda\alpha\mu},
		\end{align*}
		where in the l.h.s of the second line we have used skew-symmetry on both pair of indices, comparing its r.h.s to the one we had before we conclude
		\begin{align*}
		0&=-\overline{R}_{\mu\alpha\lambda\beta}+\overline{R}_{\beta\lambda\alpha\mu}\\
		\overline{R}_{\mu\alpha\lambda\beta}&=\overline{R}_{\lambda\beta\mu\alpha},
		\end{align*}
		we observe that in the r.h.s. we used skew-symmetry twice.
	\end{proof}
\end{proposition}
\section{The Noncommutative Einstein Tensor}
The core of general relativity is manifested in the Einstein field equations that are composed of a stress-energy momentum tensor and the Einstein tensor. Conditions imposed on the Einstein tensor, i.e. symmetry and divergenceless are important for various physical reasons, see \cite{Wald:1984rg}. In particular, these conditions stem from the properties of the energy-momentum tensor. In the following, we assume that for noncommutative spacetimes the afore-mentioned properties of the energy-momentum tensor hold as well. In particular, for this section we use the formerly constructed covariant derivatives and connections, for a Lie-algebraic spacetime, in order to construct the  (noncommutative version of) the Einstein tensor. By requiring the tensor to be divergenceless and by using the Bianchi identities and some of the curvature tensor symmetries (see   Section \ref{sec:4}) we  obtain a unique noncommutative version of the Einstein tensor. 
\begin{proposition} \label{prop:ricsym}
	The  barred Ricci tensor is symmetric, i.e.
	\begin{align*}
	\overline{R}_{\mu\nu}=\overline{R}_{\nu\mu}.
	\end{align*}
	\begin{proof}
		Take the \textit{first Bianchi identity}, see Proposition \ref{prop:bianchi1}, and trace over the first and third indices
		\begin{align*}
		0&=\overline{R}^\lambda_{\;\;[\mu\lambda\nu]}=\overline{R}^\lambda_{\;\;\mu\lambda\nu}+\overline{R}^\lambda_{\;\;\nu\mu\lambda}+\overline{R}^\lambda_{\;\;\lambda\nu\mu}=\overline{R}_{\mu\nu}-\overline{R}_{\nu\mu}+g^{\lambda\alpha}\overline{R}_{\alpha\lambda\nu\mu}=\overline{R}_{\mu\nu}-\overline{R}_{\nu\mu},
		\end{align*}
		we have used the fact that the Riemann tensor is skew-symmetric in its two first indices.
	\end{proof}
\end{proposition}	
\begin{theorem}
	The divergenceless  and symmetric \textbf{quantum    Einstein tensor} is given by \begin{align}
	G_{\sigma\lambda}=\tilde{R}_{\sigma\lambda}   
	- \frac{1}{2}g_{\sigma\lambda} \tilde{R}+g_{\sigma\lambda}\,\Lambda 
	+\frac{1}{2}	{K}_{ \sigma\lambda},
	\end{align}
	where $\Lambda\in\mathbb{R}$ and we   defined $K_{\sigma\lambda}:=g^{\kappa\alpha}K_{\kappa\sigma\alpha\lambda}$.  
\begin{proof} 
The second Bianchi identity is important for general relativity because of its contractions. First, let us  expand
the identity and contract $\mu$ and $\alpha$ 
	\begin{align*}0&=
\nabla_{\kappa} 	\tilde{R}^{\mu}_{\;\;\beta\alpha\lambda} 
+\nabla_{\lambda} 	\tilde{R}^{\mu}_{\;\;\beta\kappa\alpha} +\nabla_{\alpha} 	\tilde{R}^{\mu}_{\;\;\beta\lambda\kappa}  \\&=
\nabla_{\kappa} 	\tilde{R}^{\alpha}_{\;\;\beta\alpha\lambda} 
+\nabla_{\lambda} 	\tilde{R}^{\alpha}_{\;\;\beta\kappa\alpha} +\nabla_{\alpha} 	\tilde{R}^{\alpha}_{\;\;\beta\lambda\kappa}. 
\end{align*}  
Using the definition of the Ricci tensor and the skew-symmetry of the Riemann tensor we have
	\begin{align*}0& =
\nabla_{\kappa} 	\tilde{R}_{\beta\lambda}  +\nabla_{\alpha} 	\tilde{R}^{\alpha}_{\;\;\beta\lambda\kappa} 
-\nabla_{\lambda} 	\tilde{R}_{\beta\kappa}
\end{align*} 
Next, we contract on $\beta$ and $\kappa$ using the metric 
\begin{align*}0& =g^{\beta\kappa}
\nabla_{\kappa} 	\tilde{R}_{\beta\lambda}  +g^{\beta\kappa}\nabla_{\alpha} 	\tilde{R}^{\alpha}_{\;\;\beta\lambda\kappa} 
-g^{\beta\kappa}\nabla_{\lambda} 	\tilde{R}_{\beta\kappa}
\\&= 
\nabla_{\kappa} (g^{\beta\kappa}	\tilde{R}_{\beta\lambda} ) +g^{\beta\kappa}\nabla_{\alpha} 	\tilde{R}^{\alpha}_{\;\;\beta\lambda\kappa} 
-\nabla_{\lambda} 	(g^{\beta\kappa}\tilde{R}_{\beta\kappa})
\\&= 
\nabla_{\kappa} ( 	\tilde{R}^{\kappa}_{\;\;\lambda} ) +\nabla_{\alpha} (	g^{\beta\kappa}\tilde{R}^{\alpha}_{\;\;\beta\lambda\kappa} )
-\nabla_{\lambda} \tilde{R}.
\end{align*}
We  calculate the second term next by  using the symmetries of the last proposition,
\begin{align*} 
\nabla_{\alpha}(	g^{\beta\kappa}\tilde{R}^{\alpha}_{\;\;\beta\lambda\kappa} )&=
\nabla^{\alpha}(	g^{\beta\kappa}\tilde{R}_{\alpha\beta\lambda\kappa} )\\&
= -
\nabla^{\alpha}(	g^{\beta\kappa}\tilde{R}_{\alpha\beta\kappa\lambda} )\\&
= -
\nabla^{\sigma}(	g^{\kappa\alpha}\tilde{R}_{\sigma\kappa\alpha\lambda} )\\&
= -
\nabla^{\sigma}(	g^{\kappa\alpha}( -\tilde{R}_{\kappa\sigma\alpha\lambda}
-K_{\kappa\sigma\alpha\lambda} ) )\\&
=
\nabla^{\sigma}(	g^{\kappa\alpha}  \tilde{R}_{\kappa\sigma\alpha\lambda})
+
\nabla^{\sigma}(	g^{\kappa\alpha}K_{\kappa\sigma\alpha\lambda}  ).
\end{align*} 
By reinserting it into the second Bianchi identity 
and by using the definition of $K_{\sigma\lambda}$
we have 
	\begin{align*}0&=
\nabla_{\kappa} ( 	\tilde{R}^{\kappa}_{\;\;\lambda} ) +\nabla_{\kappa}(	   \tilde{R}^{ \kappa}_{\;\;\lambda})
+
\nabla_{\kappa}(	  {K}^{ \kappa}_{\;\;\lambda} )
-\nabla_{\lambda} \tilde{R}\\&=2\nabla_{\kappa} ( 	\tilde{R}^{\kappa}_{\;\;\lambda} )  
+
\nabla_{\kappa}(	  {K}^{ \kappa}_{\;\;\lambda} )
-\nabla_{\lambda} \tilde{R}
\\&=\nabla_{\kappa} \left(  2	\tilde{R}^{\kappa}_{\;\;\lambda}   
+
 	  {K}^{ \kappa}_{\;\;\lambda}  
- \delta^{\kappa}_{\lambda} \tilde{R}\right),
\end{align*}
by raising index we have  
	\begin{align*}0&=
	\nabla^{\sigma} \left(   	\tilde{R}_{\sigma\lambda}   
		- \frac{1}{2}g_{\sigma\lambda} \tilde{R}
+\frac{1}{2}	{K}_{ \sigma\lambda}  
\right), 
\end{align*}   
However, as in the commutative case, the additional term of a constant multiplied by the metric  can be added without changing symmetry or divergencelessness. Next, we prove that the Einstein tensor is symmetric, i.e. $G_{\sigma\lambda}=G_{\lambda\sigma}$. The term with the curvature scalar and the metric is obviously symmetric. Hence, we concentrate only on the term 	\begin{align*}\tilde{R}_{\sigma\lambda}    
+\frac{1}{2}	{K}_{ \sigma\lambda}&=  
	\overline{R}_{\sigma \lambda}+g^{\alpha\kappa}\biggl( \nabla_{\lambda}  \Sigma_{\alpha\sigma\kappa} -\nabla_{\alpha}  \Sigma_{\lambda\sigma\kappa}   \biggr)
	\\&+\frac{1}{2}g^{\alpha\kappa}\biggl(\nabla_{\alpha} \Sigma_{\lambda\kappa\sigma} +  \nabla_{\alpha} \Sigma_{\lambda\sigma\kappa}
	-\nabla_{\lambda} \Sigma_{\alpha\kappa\sigma} 
	- \nabla_{\lambda} \Sigma_{\alpha\sigma\kappa} \biggr)
	\\&=
	\overline{R}_{\sigma \lambda} +\frac{1}{2}g^{\alpha\kappa}\biggl(\nabla_{\alpha} \Sigma_{\lambda\kappa\sigma} -  \nabla_{\alpha} \Sigma_{\lambda\sigma\kappa}	+ \nabla_{\lambda} \Sigma_{\alpha\sigma\kappa} 
	-\nabla_{\lambda} \Sigma_{\alpha\kappa\sigma} \biggr)	\\&=
	\overline{R}_{\sigma \lambda} +\frac{1}{2}g^{\alpha\kappa}\biggl(\nabla_{\alpha} U_{\lambda\kappa\sigma}  	+ \nabla_{\lambda} U_{\alpha\sigma\kappa}
\biggr)\\&=
\overline{R}_{\sigma \lambda} +\frac{1}{2}g^{\alpha\kappa}\biggl(\nabla_{\alpha} U_{\lambda\kappa\sigma}  	- \nabla_{\lambda} U_{\alpha\kappa\sigma}
\biggr)\\&=:
\overline{R}_{\sigma \lambda} +\frac{1}{2}L_{\sigma \lambda},
\end{align*}  
where in the last lines we used the split property, 
$$\tilde{R}_{\sigma \lambda}=	\overline{R}_{\sigma \lambda}+g^{\alpha\nu}\biggl( \nabla_{\lambda}  \Sigma_{\alpha\sigma\nu} -\nabla_{\alpha}  \Sigma_{\lambda\sigma\nu}   \biggr),$$
and defined,
\begin{align*}
U_{\lambda\kappa\sigma} :=  \Sigma_{\lambda\kappa\sigma} -  \Sigma_{\lambda\sigma\kappa},   \ \ \  \ \ 	L_{\sigma \lambda}:=g^{\alpha\kappa}\biggl(\nabla_{\alpha} U_{\lambda\kappa\sigma}  	- \nabla_{\lambda} U_{\alpha\kappa\sigma}
\biggr).
\end{align*}   
From Proposition (\ref{prop:ricsym}) we know that the barred Ricci tensor $\overline{R}_{\sigma \lambda}$ is symmetric. Hence, all that  is left in order to prove that the Einstein tensor is symmetric is to prove that the tensor $L$ is symmetric. By splitting the tensor $L$ in a skew-symmetric and symmetric part we have,  
\begin{align*}
L_{\sigma \lambda}&= \frac{1}{2}L^+_{\sigma \lambda}
+\frac{1}{2}L^-_{\sigma \lambda} 
\end{align*}  
Next, we prove that the skew-symmetric part vanishes, i.e.\
\begin{align*}
 L^-_{\sigma \lambda} &=g^{\alpha\kappa}\biggl(\nabla_{\alpha} U_{\lambda\kappa\sigma}  
 -\nabla_{\alpha} U_{\sigma\kappa\lambda}  
 + \nabla_{\sigma} U_{\alpha\kappa\lambda}	- \nabla_{\lambda} U_{\alpha\kappa\sigma}
 \biggr)\\&=
 g^{\alpha\kappa}\biggl(\nabla_{\alpha} U_{\lambda\kappa\sigma}  
 -\nabla_{\alpha} U_{\sigma\kappa\lambda}  
 + \nabla_{\kappa} U_{\alpha\sigma\lambda} 
 \biggr)\\&=
 g^{\alpha\kappa}\biggl(\nabla_{\alpha} U_{\lambda\kappa\sigma}  
 -\nabla_{\alpha} U_{\sigma\kappa\lambda}  
 + \nabla_{\alpha} U_{\kappa\sigma\lambda} 
 \biggr)\\&=
 g^{\alpha\kappa}\nabla_{\alpha}\biggl( U_{\lambda\kappa\sigma}  
 -  U_{\sigma\kappa\lambda}  
 +   U_{\kappa\sigma\lambda} 
 \biggr)\\&=
 g^{\alpha\kappa}\nabla_{\alpha}\biggl(  
 \Sigma_{\lambda\kappa\sigma}  - \Sigma_{\lambda\sigma\kappa}  
 -  \Sigma_{\sigma\kappa\lambda}  +\Sigma_{\sigma\lambda\kappa}  
 +   \Sigma_{\kappa\sigma\lambda} - \Sigma_{\kappa\lambda\sigma} 
 \biggr)\\&=0,
\end{align*} 
where in the last lines we used the second Bianchi identity (see Proposition \ref{prop:secbi})
\begin{align*}
0=
\nabla_{\beta} 	\left( 	U _{\lambda \nu\mu }\right)+	\nabla_{\mu}\left( U _{\lambda \beta\nu  }\right)+\nabla_{\nu}\left(U_{\lambda \mu\beta  } \right),
\end{align*}
and the symmetry of $\Sigma$ in the first two indices.

\end{proof}
\end{theorem}
The former theorem proves that the noncommutative analogue of the Einstein tensor is not just the straightforward generalization from commutative to noncommutative, but has an extra term we defined by the symbol $K$. The next proposition proves that the tensor $K_{\sigma\lambda}$ is traceless and therefore not proportional to the term $g_{\sigma\lambda}\,\Lambda$ for $\Lambda\neq0$. 
\begin{proposition} 
	The quantity $K_{\sigma\lambda}$ is traceless, i.e.\ $$K=g^{\sigma\lambda}K_{\sigma\lambda}=0.$$
	\begin{proof}
	From Proposition \ref{prop:riemsym} we calculate $K_{\sigma\lambda}=g^{\kappa\alpha}K_{\kappa\sigma\alpha\lambda}$ 
	\begin{align*}
		K_{\sigma\lambda}&=g^{\kappa\alpha}K_{\kappa\sigma\alpha\lambda}=g^{\kappa\alpha}\left(\nabla_{\alpha} \Sigma_{\lambda\kappa\sigma} +  \nabla_{\alpha} \Sigma_{\lambda\sigma\kappa}
		-\nabla_{\lambda} \Sigma_{\alpha\kappa\sigma} 
		- \nabla_{\lambda} \Sigma_{\alpha\sigma\kappa}\right)\\
		&=\nabla^{\alpha} \Sigma_{\lambda\alpha\sigma} +  \nabla^{\alpha} \Sigma_{\lambda\sigma\alpha}
		-\nabla_{\lambda} (g^{\kappa\alpha}\Sigma_{\alpha\kappa\sigma}) 
		- \nabla_{\lambda} (g^{\kappa\alpha}\Sigma_{\alpha\sigma\kappa}),
	\end{align*}	
	the trace is immediate from this
	\begin{align*}
		K&=g^{\sigma\lambda}K_{\sigma\lambda}=\nabla^{\alpha}(g^{\sigma\lambda}\Sigma_{\lambda\alpha\sigma}) +  \nabla^{\alpha} (g^{\sigma\lambda}\Sigma_{\lambda\sigma\alpha})
		-\nabla^{\lambda} (g^{\kappa\alpha}\Sigma_{\alpha\kappa\lambda}) 
		- \nabla^{\lambda} (g^{\kappa\alpha}\Sigma_{\alpha\lambda\kappa})\\
		&=0,
	\end{align*}
	where the terms all  cancel.
	\end{proof}
\end{proposition}	

	\section{Conclusions and Outlook}
By using the works of (\cite{C}, \cite{DV1}, \cite{DV2} and \cite{BM}) we   calculated a formula for the covariant derivative and connection for a general noncommutative spacetime of the Lie-algebraic type up to first order, see Theorem \ref{P3.1}.
The main ingredient is to demand the centrality of the metric tensor (or the line-element) and the exterior derivative $d$ that can be defined for any associative unital algebra. This is done by using the universal differential calculus. \newline\newline
The general formula allowed us to calculate the corresponding geometrical entities such as the the Riemann tensor and the Ricci scalar. This furthermore allowed us to write a unique Einstein-tensor that is symmetric and divergenceless.\newline\newline
The noncommutative Einstein tensor has an unexpected additional term, denoted by the tensor $K$. This term would be missing if one would just replace the commutative geometrical quantities of the Einstein tensor with the noncommutative versions thereof.  From a  mathematical point of view it is interesting to prove   that the noncommutative Riemann tensor given in Proposition \ref{prop:44},  satisfies the Gauss-Bonnet theorem, as is the case of the noncommutative torus of Connes \cite{Connes2}. 
\newline\newline
The next line of work in this context is investigating physical consequences of this formula and examining the physical reality related to those models. Furthermore, with all the geometry issues  settled the next steps in progress are to explore the dynamics of matter in  a noncommutative spacetime. The interest w.r.t. matter comes from the argument  that   noncommutative geometry may give origin to matter, see \cite{chamseddine_connes_1996} and the review \cite{R3}.  By applying our general formulas to specific models we intend to see if this statement holds in the context of noncommutative spacetimes of   Lie-algebraic type. This however is current work in progress.   
\section*{Acknowledgments}
The authors acknowledge partial support from DGAPA-UNAM grant IN $103919$, 
CONACYT project $237503$ and  CONACYT project  $258259$.  

\appendix
\section{Appendix}
In this appendix we present all the proofs that were omitted from the main text in order to facilitate its readability.
\subsection{Proof of Proposition \ref{prop:cmd}}\label{proof:e2}

	\begin{proof}
We first calculate the commutator of the algebra with the inverse metric. This is done by simple algebraic manipulations
\begin{align*}0=[x^{\beta},\delta_{\mu}^{\nu}]= 
[x^{\beta},g_{\mu\rho}g^{\rho\nu}] =[x^{\beta},g_{\mu\rho}] g^{\rho\nu}+g_{\mu\rho}[x^{\beta},g^{\rho\nu}] ,
\end{align*}
which leads to
\begin{align*}g^{\mu\sigma}
[x^{\beta},g_{\mu\rho}] g^{\rho\nu}=-g^{\mu\sigma}g_{\mu\rho}[x^{\beta},g^{\rho\nu}] ,
\end{align*} 
which finally leads to 
\begin{align*}
[x^{\beta},g^{\sigma\nu}] =-g^{\mu\sigma}
[x^{\beta},g_{\mu\rho}] g^{\rho\nu}
& =-g^{\mu\sigma}(
 D^{\kappa\beta}_{\ \ \mu}g_{\kappa\rho}+ D^{\kappa\beta}_{\ \ \rho}g_{\kappa\mu}
 )g^{\rho\nu}\\&=-
 (
 D^{\nu\beta}_{\ \ \mu}g^{\mu\sigma}+ D^{\sigma\beta}_{\ \ \rho} g^{\rho\nu}
 )\\&=-
 (  D^{\sigma\beta}_{\ \ \rho} g^{\rho\nu}+
 D^{\nu\beta}_{\ \ \rho}g^{\rho\sigma}
 ).
\end{align*}  
Next, we turn to the commutator of the generators of the algebra and the derivative of the metric up to first order $[x^{\mu},D_{ \alpha}g_{\lambda\nu} ]$. To calculate this expression we take the   derivative (that fulfils Leibniz) of the known commutator of the generators and the metric, i.e.
\begin{align*}
D_{ \alpha}[x^{\mu}, g_{\lambda\nu} ]&=[D_{ \alpha}x^{\mu}, g_{\lambda\nu} ]+[x^{\mu},D_{ \alpha}g_{\lambda\nu} ]\\&=D_{\ \  \alpha}^{\mu\rho}
[\hat{\partial_{\rho}} , g_{\lambda\nu} ]+[x^{\mu},D_{ \alpha}g_{\lambda\nu} ]
\\&=D_{\  \ \alpha}^{\mu\beta}
  D_{\beta} g_{\lambda\nu}  +[x^{\mu},D_{ \alpha}g_{\lambda\nu} ]\\&=
  D_{\ \ \lambda}^{\beta\mu}D_{ \alpha}g_{\beta\nu}+  D_{ \ \ \nu}^{\beta\mu}D_{ \alpha}g_{\beta\lambda},
\end{align*}  
where the last line we have from inserting the result of the commutator of the left hand side and applying the derivative.  
\end{proof}

\subsection{Proof of Proposition \ref{prop:ccd}}\label{proof:cdww1}

\begin{proof}

We split the following commutator in two terms
\begin{align*}
[x^{\kappa}, \overline{\Gamma}_{\rho\sigma}^{\mu} ]&=\frac{1}{2}
[x^{\kappa}, g^{\mu\lambda}]
\left(D_{\sigma}g_{\lambda\rho}+D_{\rho}g_{\lambda\sigma}-D_{\lambda}g_{\rho\sigma}
\right)\\&\qquad\qquad  +\frac{1}{2} g^{\mu\lambda}
\left(	[x^{\kappa},D_{\sigma}g_{\lambda\rho}]+	[x^{\kappa},D_{\rho}g_{\lambda\sigma}]-	[x^{\kappa},D_{\lambda}g_{\rho\sigma}]
\right) ,
\end{align*}
where we calculate all commutators by using  Proposition \ref{prop:cmd}
\begin{align*}
[x^{\kappa}, \overline{\Gamma}_{\rho\sigma}^{\mu} ]&=-\frac{1}{2}
\left( 
D^{\mu\kappa}_{\ \ \beta}g^{\beta\lambda}+D^{\lambda\kappa}_{\ \ \beta}g^{\beta\mu}
\right) 
\left(D_{\sigma}g_{\lambda\rho}+D_{\rho}g_{\lambda\sigma}-D_{\lambda}g_{\rho\sigma}
\right)\\&\qquad\qquad  +\frac{1}{2} g^{\mu\lambda}
\biggl(	D^{\beta\kappa}_{\ \ \lambda}D_{\sigma}g_{\beta\rho}+D^{\beta\kappa}_{\ \ \rho}D_{\sigma}g_{\beta\lambda}-D^{\kappa\beta}_{\ \ \sigma}D_{\beta}g_{\lambda\rho}\\&\qquad\qquad\qquad \qquad+
D^{\beta\kappa}_{\ \ \lambda}D_{\rho}g_{\beta\sigma}+D^{\beta\kappa}_{\ \ \sigma}D_{\rho}g_{\beta\lambda}-D^{\kappa\beta}_{\ \ \rho}D_{\beta}g_{\lambda\sigma}
\\&\qquad\qquad\qquad\qquad -
D^{\beta\kappa}_{\ \ \rho}D_{\lambda}g_{\beta\sigma}-D^{\beta\kappa}_{\ \ \sigma}D_{\lambda}g_{\beta\rho}+D^{\kappa\beta}_{\ \ \lambda}D_{\beta}g_{\rho\sigma}
\biggr).
\end{align*}
After a few straightforward cancellations we have, 
\begin{align*}
[x^{\kappa}, \overline{\Gamma}_{\rho\sigma}^{\mu} ]&=  -	D^{\mu\kappa}_{\ \ \beta}\overline{\Gamma}^{\beta}_{\rho\sigma} -\frac{1}{2}
\left( D^{\lambda\kappa}_{\ \ \beta}g^{\beta\mu}
\right) 
\left(\cancel{D_{\sigma}g_{\lambda\rho}}+\bcancel{D_{\rho}g_{\lambda\sigma}}-D_{\lambda}g_{\rho\sigma}
\right)\\&\qquad\qquad  +\frac{1}{2} g^{\mu\beta}
\biggl(\cancel{	D^{\lambda\kappa}_{\ \ \beta}D_{\sigma}g_{\lambda\rho}}+D^{\lambda\kappa}_{\ \ \rho}D_{\sigma}g_{\lambda\beta}-D^{\kappa\lambda}_{\ \ \sigma}D_{\lambda}g_{\beta\rho}\\&\qquad\qquad\qquad \qquad+\bcancel{
	D^{\lambda\kappa}_{\ \ \beta}D_{\rho}g_{\lambda\sigma}}+D^{\lambda\kappa}_{\ \ \sigma}D_{\rho}g_{\lambda\beta}-D^{\kappa\lambda}_{\ \ \rho}D_{\lambda}g_{\beta\sigma}
\\&\qquad\qquad\qquad\qquad -
D^{\lambda\kappa}_{\ \ \rho}D_{\beta}g_{\lambda\sigma}-D^{\lambda\kappa}_{\ \ \sigma}D_{\beta}g_{\lambda\rho}+D^{\kappa\lambda}_{\ \ \beta}D_{\lambda}g_{\rho\sigma}
\biggr),
\end{align*}
  due to the (skew-)symmetry properties of $D$, some terms become $C$ and $S$,
  \begin{align*}
[x^{\kappa}, \overline{\Gamma}_{\rho\sigma}^{\mu} ]&=    -	D^{\mu\kappa}_{\ \ \beta}\overline{\Gamma}^{\beta}_{\rho\sigma} + 
S^{\lambda\kappa}_{\ \ \beta}g^{\beta\mu} 
\left(  D_{\lambda}g_{\rho\sigma}
\right)\\&\qquad\qquad  +\frac{1}{2} g^{\mu\beta}
\biggl(  D^{\lambda\kappa}_{\ \ \rho}(D_{\sigma}g_{\lambda\beta}
-
D_{\beta}g_{\lambda\sigma})-D^{\kappa\lambda}_{\ \ \rho}D_{\lambda}g_{\beta\sigma}   
+(\rho\leftrightarrow\sigma)  
\biggr)\\&=    -	D^{\mu\kappa}_{\ \ \beta}\overline{\Gamma}^{\beta}_{\rho\sigma} + 
S^{\lambda\kappa}_{\ \ \beta}g^{\beta\mu} 
\left(  D_{\lambda}g_{\rho\sigma}
\right)\\&\qquad\qquad  +\frac{1}{2} g^{\mu\beta}
\biggl(  D^{\lambda\kappa}_{\ \ \rho}(D_{\sigma}g_{\lambda\beta}
-
D_{\beta}g_{\lambda\sigma})-D^{\lambda\kappa}_{\ \ \rho}D_{\lambda}g_{\beta\sigma}   +C^{\lambda\kappa}_{\ \ \rho}D_{\lambda}g_{\beta\sigma}   
+(\rho\leftrightarrow\sigma)  
\biggr) 	\\&=    -	D^{\mu\kappa}_{\ \ \beta}\overline{\Gamma}^{\beta}_{\rho\sigma} + 
S^{\lambda\kappa}_{\ \ \beta}g^{\beta\mu} 
D_{\lambda}g_{\rho\sigma}
\\&\qquad\qquad  +\frac{1}{2} g^{\mu\beta}
\biggl(  D^{\lambda\kappa}_{\ \ \rho}(D_{\sigma}g_{\lambda\beta}+D_{\lambda}g_{\beta\sigma}   
-
D_{\beta}g_{\lambda\sigma})-2D^{\lambda\kappa}_{\ \ \rho}D_{\lambda}g_{\beta\sigma}   +C^{\lambda\kappa}_{\ \ \rho}D_{\lambda}g_{\beta\sigma}   
\\&\qquad\qquad+(\rho\leftrightarrow\sigma ) 
\biggr) \\&=    -	D^{\mu\kappa}_{\ \ \lambda}\overline{\Gamma}^{\lambda}_{\rho\sigma} 	+D^{\lambda\kappa}_{\ \ \rho}\overline{\Gamma}^{\mu}_{\lambda\sigma}+D^{\lambda\kappa}_{\ \ \sigma}\overline{\Gamma}^{\mu}_{\lambda\rho}\\&\qquad\qquad
+ 
S^{\lambda\kappa}_{\ \ \beta}g^{\beta\mu} 
D_{\lambda}g_{\rho\sigma} 
+\frac{1}{2} g^{\mu\beta}
\biggl(  -2D^{\lambda\kappa}_{\ \ \rho}D_{\lambda}g_{\beta\sigma}   +C^{\lambda\kappa}_{\ \ \rho}D_{\lambda}g_{\beta\sigma}   
+(\rho\leftrightarrow\sigma)  
\biggr) \\&=    -	D^{\mu\kappa}_{\ \ \lambda}\overline{\Gamma}^{\lambda}_{\rho\sigma} 	+D^{\lambda\kappa}_{\ \ \rho}\overline{\Gamma}^{\mu}_{\lambda\sigma}+D^{\lambda\kappa}_{\ \ \sigma}\overline{\Gamma}^{\mu}_{\lambda\rho}\\&\qquad\qquad 
- g^{\mu\beta}
\biggl( S^{\lambda\kappa}_{\ \ \rho}D_{\lambda}g_{\beta\sigma}   
+S^{\lambda\kappa}_{\ \ \sigma}D_{\lambda}g_{\beta\rho} -  	S^{\lambda\kappa}_{\ \ \beta}
D_{\lambda}g_{\rho\sigma} 
\biggr) ,
\end{align*}
where in the last lines we summarized a term to a connection symbol,   replaced the indices $\beta$ and $\lambda$,
used the identities $D^{\kappa\lambda}_{\ \ \sigma}+D^{\lambda\kappa}_{\ \ \sigma}=2S^{\lambda\kappa}_{\ \ \sigma}$ and $D^{\kappa\lambda}_{\ \ \sigma}=D^{\lambda\kappa}_{\ \ \sigma}-C^{\lambda\kappa}_{\ \ \sigma}$.
\end{proof}  
\subsection{Proof of Proposition \ref{prop:cdww}}\label{proof:cdww}
		\begin{proof}
			We proceed by applying the covariant derivative on the tensor
			\begin{align*}
			\nabla(M)&=(\nabla \otimes \mathbb{I})M+(\sigma\otimes \mathbb{I})\circ(\mathbb{I}\otimes \nabla)M,
			\end{align*}
			where the first term yields
			\begin{align*}
			(\nabla \otimes \mathbb{I})M=\nabla(M_{\mu\nu}dx^\mu)\otimes  dx^\nu=D_{\alpha}(M_{\mu\nu}) dx^\alpha\otimes  dx^\mu\otimes  dx^\nu-M_{\mu\nu}\tilde{\Gamma}^\mu_{\rho\sigma}dx^\rho\otimes  dx^\sigma\otimes  dx^\nu,
			\end{align*}
			while the second is
			\begin{align*}
			(\sigma\otimes \mathbb{I})\circ(\mathbb{I}\otimes \nabla)M&=(\sigma\otimes \mathbb{I})\circ(M_{\mu\nu}dx^\mu\otimes \nabla(dx^\nu))\\
			&=-(\sigma\otimes \mathbb{I})\circ(M_{\mu\nu}dx^\mu\tilde{\Gamma}^\nu_{\alpha\beta}\otimes dx^\alpha\otimes dx^\beta)\\
			&=-(\sigma\otimes \mathbb{I})\circ(M_{\mu\nu}\tilde{\Gamma}^\nu_{\alpha\beta}dx^\mu\otimes dx^\alpha\otimes dx^\beta)\\
			&\qquad \qquad-(\sigma\otimes \mathbb{I})\circ(M_{\mu\nu}[dx^\mu,\tilde{\Gamma}^\nu_{\alpha\beta}]\otimes dx^\alpha\otimes dx^\beta)\\
			&=-\sigma(M_{\mu\nu}\tilde{\Gamma}^\nu_{\alpha\beta}dx^\mu\otimes dx^\alpha)\otimes dx^\beta-\sigma(M_{\mu\nu}\Sigma^{\mu\nu}_{\alpha\beta\lambda}dx^\lambda\otimes dx^\alpha)\otimes dx^\beta \\
			&=-M_{\mu\nu}\tilde{\Gamma}^\nu_{\alpha\beta}\sigma(dx^\mu\otimes dx^\alpha)\otimes dx^\beta-M_{\mu\nu}\Sigma^{\mu\nu}_{\alpha\beta\lambda}\sigma(dx^\lambda\otimes dx^\alpha)\otimes dx^\beta, 
			\end{align*}
			the calculation of the braiding acting on the basis renders
			\begin{align*}
			\sigma(dx^\mu\otimes dx^\alpha)&=dx^\alpha\otimes dx^\mu+[x^\alpha,\nabla(dx^\mu)]+\nabla([dx^\mu,x^\alpha])\\
			&=dx^\alpha\otimes dx^\mu-[x^\alpha,\tilde{\Gamma}^\mu_{\rho\sigma}dx^\rho\otimes dx^\sigma]-D^{\mu\alpha}_{\ \ \lambda}\tilde{\Gamma}^\lambda_{\rho\sigma}dx^\rho\otimes dx^\sigma\\
			&=dx^\alpha\otimes dx^\mu-([x^\alpha,\tilde{\Gamma}^\mu_{\rho\sigma}]+D^{\mu\alpha}_{\ \ \lambda}\tilde{\Gamma}^\lambda_{\rho\sigma})dx^\rho\otimes dx^\sigma-\tilde{\Gamma}^\mu_{\rho\sigma}[x^\alpha,dx^\rho\otimes dx^\sigma]\\
			&=dx^\alpha\otimes dx^\mu-([x^\alpha,\tilde{\Gamma}^\mu_{\rho\sigma}]+D^{\mu\alpha}_{\ \ \lambda}\tilde{\Gamma}^\lambda_{\rho\sigma})dx^\rho\otimes dx^\sigma\\
			&\qquad \qquad\qquad +\tilde{\Gamma}^\mu_{\rho\sigma}(D^{\rho\alpha}_{\ \ \lambda} dx^\lambda\otimes dx^\sigma+D^{\sigma\alpha}_{\ \ \lambda} dx^\rho\otimes dx^\lambda)\\
			&=dx^\alpha\otimes dx^\mu-([x^\alpha,\tilde{\Gamma}^\mu_{\rho\sigma}]+D^{\mu\alpha}_{\ \ \eta}\tilde{\Gamma}^\eta_{\rho\sigma}-D^{\eta\alpha}_{\ \ \rho}\tilde{\Gamma}^\mu_{\eta\sigma}-D^{\eta\alpha}_{\ \ \sigma}\tilde{\Gamma}^\mu_{\rho\eta})dx^\rho\otimes dx^\sigma.
			\end{align*}
		Next, we incorporate our result for $[x^\alpha,\tilde{\Gamma}^\mu_{\rho\sigma}]$ from Proposition \ref{prop:ccd} and Equation (\ref{eq:ccd_eq1}) to find
			\begin{align*}
			\sigma(dx^\mu\otimes dx^\alpha)
			&=dx^\alpha\otimes dx^\mu+g^{\mu\beta}
			\biggl( S^{\lambda\alpha}_{ \ \ \rho}D_{\lambda}g_{\beta\sigma}   
			+S^{\lambda\alpha}_{\ \ \sigma}D_{\lambda}g_{\beta\rho} -  	S^{\lambda\alpha}_{\ \ \beta}
			D_{\lambda}g_{\rho\sigma} 
			\biggr) dx^\rho\otimes dx^\sigma.
			\end{align*}
			By inserting this term into the previous calculation we obtain
			\begin{align*}
			(\sigma\otimes \mathbb{I})\circ(\mathbb{I}\otimes \nabla)M
			&=-M_{\mu\nu}\tilde{\Gamma}^\nu_{\alpha\beta}\left(\delta^\alpha_\rho\delta^\mu_\sigma+g^{\mu\tau}
			\biggl( S^{\lambda\alpha}_{ \ \ \rho}D_{\lambda}g_{\tau\sigma}   
			+S^{\lambda\alpha}_{\ \ \sigma}D_{\lambda}g_{\tau\rho} \right.\\
			&\left.-  	S^{\lambda\alpha}_{\ \ \tau}
			D_{\lambda}g_{\rho\sigma} 
			\biggr) \right)dx^\rho\otimes dx^\sigma\otimes dx^\beta
			-M_{\mu\nu}\Sigma^{\mu\nu}_{\alpha\beta\lambda}\left(\delta^\alpha_\rho\delta^\lambda_\sigma+g^{\lambda\tau}
			\biggl( S^{\theta\alpha}_{\ \ \rho}D_{\theta}g_{\tau\sigma} \right.\\
			&\left.  
			+S^{\theta\alpha}_{\ \ \sigma}D_{\theta}g_{\tau\rho} -  	S^{\theta\alpha}_{\ \ \tau}
			D_{\theta}g_{\rho\sigma} 
			\biggr) \right)dx^\rho\otimes dx^\sigma\otimes dx^\beta, 
			\end{align*}
			where up to first order we have
			\begin{align*}
			(\sigma\otimes \mathbb{I})\circ(\mathbb{I}\otimes \nabla)M
			&=-M_{\mu\nu}\tilde{\Gamma}^\nu_{\alpha\beta}\left(\delta^\alpha_\rho\delta^\mu_\sigma+g^{\mu\tau}
			\biggl( S^{\lambda\alpha}_{\ \ \rho}D_{\lambda}g_{\tau\sigma}   
			+S^{\lambda\alpha}_{\ \ \sigma}D_{\lambda}g_{\tau\rho} \right.\\
			&\left.-  	S^{\lambda\alpha}_{\ \ \tau}
			D_{\lambda}g_{\rho\sigma} 
			\biggr) \right)dx^\rho\otimes dx^\sigma\otimes dx^\beta- M_{\mu\nu}\Sigma^{\mu\nu}_{\ \ \rho\beta\sigma}dx^\rho\otimes dx^\sigma\otimes dx^\beta.
			\end{align*}  Joining both the first and the second term we get
			\begin{align*}
			\nabla (M)&=\left(D_{\rho}(M_{\sigma\nu})-M_{\lambda\nu}\tilde{\Gamma}^\lambda_{\rho\sigma}-M_{\sigma\lambda}\tilde{\Gamma}^\lambda_{\rho\nu}\right) dx^\rho\otimes  dx^\sigma\otimes  dx^\nu\\
			&-M_{\mu\gamma}\Gamma^\gamma_{\alpha\nu}g^{\mu\tau}
			\biggl( S^{\lambda\alpha}_{\ \ \rho}D_{\lambda}g_{\tau\sigma}   
			+S^{\lambda\alpha}_{\ \ \sigma}D_{\lambda}g_{\tau\rho} -  	S^{\lambda\alpha}_{\ \ \tau}
			D_{\lambda}g_{\rho\sigma} 
			\biggr) dx^\rho\otimes dx^\sigma\otimes dx^\nu\\
			&- M_{\mu\kappa}\Sigma^{\mu\kappa}_{\rho\nu\sigma}  dx^\rho\otimes dx^\sigma\otimes dx^\nu.
			\end{align*}
			Therefore, if we put everything into indices we have
			\begin{align*}
			\nabla_\rho M_{\mu\nu}&=D_{\rho}(M_{\mu\nu})-M_{\lambda\nu}\tilde{\Gamma}^\lambda_{\rho\mu}-M_{\mu\lambda}\tilde{\Gamma}^\lambda_{\rho\nu}\\
			&
			-M_{\eta\gamma}g^{\eta\tau}\Gamma^\gamma_{\alpha\nu}
			\biggl( S^{\lambda\alpha}_{\ \ \rho}D_{\lambda}g_{\tau\mu}   
			+S^{\lambda\alpha}_{\ \ \mu}D_{\lambda}g_{\tau\rho} -  	S^{\lambda\alpha}_{\ \ \tau}
			D_{\lambda}g_{\rho\mu} 
			\biggr) \\
			&- M_{\eta\gamma}
			g^{\eta\tau}\Gamma^{\gamma}_{\alpha\mu}\biggl( S^{\lambda\alpha}_{\ \ \nu}D_{\lambda}g_{\tau\rho}   
			+S^{\lambda\alpha}_{\ \ \rho}D_{\lambda}g_{\tau\nu} -  	S^{\lambda\alpha}_{\ \ \tau}
			D_{\lambda}g_{\nu\rho} 
			\biggr) 
			\\
			&=D_{\rho}(M_{\mu\nu})-M_{\lambda\nu}\tilde{\Gamma}^\lambda_{\rho\mu}-M_{\mu\lambda}\tilde{\Gamma}^\lambda_{\rho\nu}
			-M_{\lambda\beta}	\Sigma^{\lambda\beta}_{\mu\rho\nu} - M_{\lambda\beta}	\Sigma^{\lambda\beta}_{\nu\rho\mu},
			\end{align*}
			where we inserted the $\Sigma$-term
		\begin{align*}
	\Sigma^{\lambda\beta}_{\alpha\nu\mu} =
	g^{\lambda\kappa}\Gamma^\beta_{\rho\mu}
	\biggl( S^{\sigma\rho}_{\ \ \alpha}D_{\sigma}g_{\kappa\nu}   
	+S^{\sigma\rho}_{\ \ \nu}D_{\sigma}g_{\kappa\alpha} -  	S^{\sigma\rho}_{\ \ \kappa}
	D_{\sigma}g_{\alpha\nu} 
	\biggr),
	\end{align*} in the last lines. We proceed by applying the covariant derivative on the tensor $T\in\Omega^1(\mathcal{A})\otimes \Omega^1(\mathcal{A})\otimes \Omega^1(\mathcal{A})$,
		\begin{align*}
		\nabla(T)=\nabla(T_{\rho\mu\nu}dx^{\rho}\otimes dx^{\mu}\otimes dx^{\nu} ).
		\end{align*}
		First, let us write the tensor $T$ as a tensor product of a one-form and a two-form (which is generated by juxtaposition), i.e. 
		\begin{align*}
		\nabla(T)= \nabla (\omega \otimes \omega'),
		\end{align*}
		where $\omega\in\Omega^1(\mathcal{A})$ and $\omega'\in\Omega^1(\mathcal{A})\otimes\Omega^1(\mathcal{A})$. By using Definition \ref{c1} the last expression is given by,  
		\begin{align*}
		\nabla(T)&= \nabla (\omega \otimes \omega') \\&=
		\nabla\omega \otimes \omega'+\sigma_3(\omega \otimes\nabla\omega').
		\end{align*}	
		Before calculating the whole expression let us first calculate explicitly the term $\nabla\omega'$, since $\omega'\in\Omega^1(\mathcal{A})\otimes\Omega^1(\mathcal{A})$
		the covariant derivative thereof is not trivial. Hence, we have 		
		\begin{align*}
		\nabla\omega'&=(\nabla \otimes \mathbb{I})\omega'+(\sigma\otimes \mathbb{I})\circ(\mathbb{I}\otimes \nabla)\omega'
		\\&= 
		\nabla  (dx^{\mu})\otimes dx^{\nu} 
		+(\sigma\otimes \mathbb{I})(dx^{\mu} \otimes  \nabla( dx^{\nu})) \\&= -\tilde{\Gamma}^{\mu}_{\lambda\kappa}
		dx^{\lambda}\otimes dx^{\kappa} \otimes dx^{\nu} 
		-(\sigma\otimes \mathbb{I})(dx^{\mu} \otimes   
		\tilde{\Gamma}^{\nu}_{\lambda\kappa}
		dx^{\lambda}\otimes dx^{\kappa} 
		).
		\end{align*}
		Next, we commute $\tilde{\Gamma}$  in order to have it to the leftmost, this   yields a $\Sigma$-term
		\begin{align*}
		\nabla\omega'&= -\tilde{\Gamma}^{\mu}_{\lambda\kappa}
		dx^{\lambda}\otimes dx^{\kappa} \otimes dx^{\nu} 
		-(\sigma\otimes \mathbb{I})((
		\Sigma^{\mu\nu}_{\lambda\kappa\sigma}dx^{\sigma} 
		+\tilde{\Gamma}^{\nu}_{\lambda\kappa} dx^{\mu})\otimes    
		dx^{\lambda}\otimes dx^{\kappa} 
		)\\&= -\tilde{\Gamma}^{\mu}_{\lambda\kappa}
		dx^{\lambda}\otimes dx^{\kappa} \otimes dx^{\nu} 
		- \Sigma^{\mu\nu}_{\lambda\kappa\sigma}\sigma ( 
		dx^{\sigma} \otimes    
		dx^{\lambda})\otimes dx^{\kappa} 
		-\tilde{\Gamma}^{\nu}_{\lambda\kappa} \sigma(dx^{\mu} \otimes    
		dx^{\lambda})\otimes dx^{\kappa}  \\&= -\tilde{\Gamma}^{\mu}_{\lambda\sigma}
		dx^{\lambda}\otimes dx^{\sigma} \otimes dx^{\nu} 
		- \left(\Sigma^{\mu\nu}_{\lambda\kappa\sigma} 
		dx^{\lambda} \otimes    
		dx^{\sigma} 
		+\tilde{\Gamma}^{\nu}_{\lambda\kappa}   dx^{\lambda} \otimes    
		dx^{\mu} \right)\otimes dx^{\kappa}  
		\\&-  {\Gamma}^{\nu}_{\tau\kappa}  g^{\mu\beta}
		\biggl( S^{\theta\tau}_{\ \ \lambda}D_{\theta}g_{\beta\sigma}   
		+S^{\theta\tau}_{ \ \ \sigma}D_{\theta}g_{\beta\lambda} -  	S^{\theta\tau}_{\ \ \beta}
		D_{\theta}g_{\lambda\sigma} 
		\biggr) dx^\lambda\otimes dx^\sigma\otimes dx^{\kappa}  
		\\&=:K^{\mu\nu}_{\lambda\sigma\kappa}dx^\lambda\otimes dx^\sigma\otimes dx^{\kappa}, 
		\end{align*}	
		where in the last lines we used the expression 	
		\begin{align*}
		\sigma(dx^\mu\otimes dx^\lambda)
		&=dx^\lambda\otimes dx^\mu+g^{\mu\beta}
		\biggl( S^{\theta\lambda}_{\ \ \alpha}D_{\theta}g_{\beta\sigma}   
		+S^{\theta\lambda}_{\ \ \sigma}D_{\theta}g_{\beta\alpha} -  	S^{\theta\lambda}_{\ \ \beta}
		D_{\theta}g_{\alpha\sigma} 
		\biggr) dx^\alpha\otimes dx^\sigma,
		\end{align*}
		and defined the object $K$ as 
		\begin{align}\label{eq:k}
		K^{\mu\nu}_{\lambda\sigma\kappa}:&= -\tilde{\Gamma}^{\mu}_{\lambda\sigma}\delta^{\nu}_{\kappa} -\tilde{\Gamma}^{\nu}_{\lambda\kappa}\delta^{\mu}_{\sigma}
		- \Sigma^{\mu\nu}_{\lambda\kappa\sigma} -  {\Gamma}^{\nu}_{\tau\kappa}  g^{\mu\beta}
		\biggl( S^{\theta\tau}_{\ \ \lambda}D_{\theta}g_{\beta\sigma}   
		+S^{\theta\tau}_{\ \ \sigma}D_{\theta}g_{\beta\lambda} -  	S^{\theta\tau}_{\ \ \beta}
		D_{\theta}g_{\lambda\sigma} 
		\biggr).
		\end{align}
		With this, we can turn our attention to the term $\sigma_3(\omega \otimes\nabla\omega')$, we start by substituting our recent findings
		\begin{align*}
		\sigma_3(\omega \otimes\nabla\omega')&=
		\sigma_3(T_{\rho\mu\nu}dx^{\rho}\otimes
		K^{\mu\nu}_{\lambda\sigma\kappa}dx^\lambda\otimes dx^\sigma\otimes dx^{\kappa} 
		)\\
		&=\sigma_3(T_{\rho\mu\nu}dx^{\rho}K^{\mu\nu}_{\lambda\sigma\kappa}\otimes dx^\lambda\otimes dx^\sigma\otimes dx^{\kappa} 
		)\\&=\sigma_3(T_{\rho\mu\nu}([dx^{\rho},K^{\mu\nu}_{\lambda\sigma\kappa}]+K^{\mu\nu}_{\lambda\sigma\kappa}dx^{\rho})\otimes dx^\lambda\otimes dx^\sigma\otimes dx^{\kappa} 
		),\\&=\sigma_3(T_{\rho\mu\nu}(-\left(\Sigma^{\rho\mu}_{\lambda\sigma\alpha}\delta^{\nu}_{\kappa} 
		+\Sigma^{\rho\nu}_{\lambda\kappa\alpha}\delta^{\mu}_{\sigma}
		\right)dx^{\alpha}+K^{\mu\nu}_{\lambda\sigma\kappa}dx^{\rho})\otimes dx^\lambda\otimes dx^\sigma\otimes dx^{\kappa} 
		),\end{align*}
		where we have used the commutator of $dx$ and $K$ which can be easily calculated up to leading order from Equation (\ref{eq:k}), this renders
		\begin{align*}
			[dx^{\rho},K^{\mu\nu}_{\lambda\sigma\kappa}]&=
			-[dx^{\rho},\tilde{\Gamma}^{\mu}_{\lambda\sigma}\delta^{\nu}_{\kappa}] -[dx^{\rho},\tilde{\Gamma}^{\nu}_{\lambda\kappa}\delta^{\mu}_{\sigma}]\\&=
			-\left(\Sigma^{\rho\mu}_{\lambda\sigma\alpha}\delta^{\nu}_{\kappa} 
			+\Sigma^{\rho\nu}_{\lambda\kappa\alpha}\delta^{\mu}_{\sigma}
			\right)dx^{\alpha}.
		\end{align*}
	By going back to the main calculation we have, 
		\begin{align*}
		\sigma_3(\omega \otimes\nabla\omega')&=\sigma_3(-T_{\rho\mu\nu} \left(\Sigma^{\rho\mu}_{\lambda\sigma\alpha}\delta^{\nu}_{\kappa} 
		+\Sigma^{\rho\nu}_{\lambda\kappa\alpha}\delta^{\mu}_{\sigma}
		\right)dx^{\alpha}\otimes dx^\lambda\otimes dx^\sigma\otimes dx^{\kappa} 
		\\
		&+T_{\rho\mu\nu}K^{\mu\nu}_{\lambda\sigma\kappa}dx^{\rho} \otimes dx^\lambda\otimes dx^\sigma\otimes dx^{\kappa} 
		)\\&= -T_{\rho\mu\nu} \left(\Sigma^{\rho\mu}_{\lambda\sigma\alpha}\delta^{\nu}_{\kappa} 
		+\Sigma^{\rho\nu}_{\lambda\kappa\alpha}\delta^{\mu}_{\sigma}
		\right)dx^{\lambda}\otimes dx^\alpha\otimes dx^\sigma\otimes dx^{\kappa} 
		\\&+T_{\rho\mu\nu}K^{\mu\nu}_{\lambda\sigma\kappa}\sigma(dx^{\rho} \otimes dx^\lambda)\otimes dx^\sigma\otimes dx^{\kappa} \\&= -T_{\rho\mu\nu} \left(\Sigma^{\rho\mu}_{\lambda\sigma\alpha}\delta^{\nu}_{\kappa} 
		+\Sigma^{\rho\nu}_{\lambda\kappa\alpha}\delta^{\mu}_{\sigma}
		\right)dx^{\lambda}\otimes dx^\alpha\otimes dx^\sigma\otimes dx^{\kappa} 
		\\&+T_{\alpha\mu\nu}K^{\mu\nu}_{\lambda\sigma\kappa}  dx^{\lambda } \otimes dx^\alpha  \otimes dx^\sigma\otimes dx^{\kappa} \\&-
		T_{\rho\mu\nu}(  {\Gamma}^{\mu}_{\tau\sigma} \delta^{\nu}_{\kappa} + {\Gamma}^{\nu}_{\tau\kappa} \delta^{\mu}_{\sigma})
		g^{\rho\beta}
		\biggl( S^{\theta\tau}_{\ \ \lambda}D_{\theta}g_{\beta\alpha}   
		+S^{\theta\tau}_{\ \ \alpha}D_{\theta}g_{\beta\lambda} \\
		&-  	S^{\theta\tau}_{\ \ \beta}
		D_{\theta}g_{\lambda\alpha} 
		\biggr) dx^\lambda\otimes dx^\alpha\otimes dx^\sigma\otimes dx^{\kappa}.
		\end{align*}
		In the following we factorize the basis of $\Omega^1(\mathcal{A})\otimes \Omega^1(\mathcal{A})\otimes \Omega^1(\mathcal{A})$ as 
		\begin{align*}
		\sigma_3(\omega \otimes\nabla\omega')&= -\biggl(T_{\rho\mu\nu} \left(\Sigma^{\rho\mu}_{\lambda\sigma\alpha}\delta^{\nu}_{\kappa} 
		+\Sigma^{\rho\nu}_{\lambda\kappa\alpha}\delta^{\mu}_{\sigma}
		\right) \\&
		+T_{\alpha\mu\nu}\biggl( +\tilde{\Gamma}^{\mu}_{\lambda\sigma}\delta^{\nu}_{\kappa} +\tilde{\Gamma}^{\nu}_{\lambda\kappa}\delta^{\mu}_{\sigma}
		+ \Sigma^{\mu\nu}_{\lambda\kappa\sigma} + {\Gamma}^{\nu}_{\tau\kappa}  g^{\mu\beta}
		\biggl( S^{\theta\tau}_{\ \ \lambda}D_{\theta}g_{\beta\sigma}   
		+S^{\theta\tau}_{\ \ \sigma}D_{\theta}g_{\beta\lambda}   \\&-  	S^{\theta\tau}_{\ \ \beta}
		D_{\theta}g_{\lambda\sigma} 
		\biggr)\biggr)+
		T_{\rho\mu\nu}(  {\Gamma}^{\mu}_{\tau\sigma} \delta^{\nu}_{\kappa} + {\Gamma}^{\nu}_{\tau\kappa} \delta^{\mu}_{\sigma})
		g^{\rho\beta}
		\biggl( S^{\theta\tau}_{\ \ \lambda}D_{\theta}g_{\beta\alpha}   
		+S^{\theta\tau}_{\ \ \alpha}D_{\theta}g_{\beta\lambda} \\&-  	S^{\theta\tau}_{\ \ \beta}
		D_{\theta}g_{\lambda\alpha} 
		\biggr) \biggr)dx^\lambda\otimes dx^\alpha\otimes dx^\sigma\otimes dx^{\kappa} 
		\\&= -\biggl(  T_{\rho\nu\kappa}\Sigma^{\rho\nu}_{\lambda\sigma\alpha} 
		+T_{\rho\sigma\nu}\Sigma^{\rho\nu}_{\lambda\kappa\alpha}+T_{\alpha\rho\nu} \Sigma^{\rho\nu}_{\lambda\kappa\sigma}
		+ T_{\alpha\mu\kappa} \tilde{\Gamma}^{\mu}_{\lambda\sigma}  +T_{\alpha\sigma\nu}\tilde{\Gamma}^{\nu}_{\lambda\kappa} 
		\\&+
		( T_{\rho\nu\kappa} {\Gamma}^{\nu}_{\tau\sigma}   +T_{\rho\sigma\nu} {\Gamma}^{\nu}_{\tau\kappa}   )
		g^{\rho\beta}
		\biggl( S^{\theta\tau}_{\ \ \lambda}D_{\theta}g_{\beta\alpha}   
		+S^{\theta\tau}_{\ \ \alpha}D_{\theta}g_{\beta\lambda} -  	S^{\theta\tau}_{\ \ \beta}
		D_{\theta}g_{\lambda\alpha} 
		\biggr) 
		\\&
		+T_{\alpha\rho\nu}  {\Gamma}^{\nu}_{\tau\kappa}  g^{\rho\beta}
		\biggl( S^{\theta\tau}_{\ \ \lambda}D_{\theta}g_{\beta\sigma}   
		+S^{\theta\tau}_{\ \ \sigma}D_{\theta}g_{\beta\lambda} \\
		&-  	S^{\theta\tau}_{\ \ \beta}
		D_{\theta}g_{\lambda\sigma} 
		\biggr) 
		\biggr)   dx^\lambda\otimes dx^\alpha\otimes dx^\sigma\otimes dx^{\kappa}  . 
		\end{align*}		
		Next, we look at the first term of the covariant derivative of the three-tensor, i.e.  
		\begin{align*}
		\nabla\omega \otimes \omega'&=
		\nabla(T_{\rho\mu\nu}dx^{\rho})\otimes dx^{\mu}\otimes dx^{\nu} \\&=D_{\lambda}T_{\alpha\sigma\kappa}\,dx^{\lambda} \otimes dx^{\alpha} \otimes dx^{\sigma}\otimes dx^{\kappa}-T_{\rho\sigma\kappa}\tilde{\Gamma}^{\rho}_{\lambda\alpha}
		dx^{\lambda} \otimes dx^{\alpha}\otimes dx^{\sigma}\otimes dx^{\kappa}.
		\end{align*} 
		By summing all the above terms the final result for the covariant derivative of the three-tensor  
		which finally obtain
		\begin{align*}
		\nabla_{\lambda}(T_{\alpha\sigma\kappa})&=   D_{\lambda}T_{\alpha\sigma\kappa}  
		-
		T_{\rho\sigma\kappa}\tilde{\Gamma}^{\rho}_{\lambda\alpha}
		-T_{\alpha\rho\kappa} \tilde{\Gamma}^{\rho}_{\lambda\sigma}  -T_{\alpha\sigma\rho}\tilde{\Gamma}^{\rho}_{\lambda\kappa}   -T_{\rho\nu\kappa}\Sigma^{\rho\nu}_{\lambda\sigma\alpha} 
		-T_{\rho\sigma\nu}\Sigma^{\rho\nu}_{\lambda\kappa\alpha}-T_{\alpha\rho\nu} \Sigma^{\rho\nu}_{\lambda\kappa\sigma}
		\\&-
		( T_{\rho\nu\kappa} {\Gamma}^{\nu}_{\tau\sigma}   +T_{\rho\sigma\nu} {\Gamma}^{\nu}_{\tau\kappa}   )
		g^{\rho\beta}
		\left( S^{\theta\tau}_{\ \ \lambda}D_{\theta}g_{\beta\alpha}   
		+S^{\theta\tau}_{\ \ \alpha}D_{\theta}g_{\beta\lambda} -  	S^{\theta\tau}_{\ \ \beta}
		D_{\theta}g_{\lambda\alpha} 
		\right)   \\&
		-T_{\alpha\rho\nu}  {\Gamma}^{\nu}_{\tau\kappa}  g^{\rho\beta}
		\biggl( S^{\theta\tau}_{\ \ \lambda}D_{\theta}g_{\beta\sigma}   
		+S^{\theta\tau}_{\ \ \sigma}D_{\theta}g_{\beta\lambda} -  	S^{\theta\tau}_{\ \ \beta}
		D_{\theta}g_{\lambda\sigma} 
		\biggr) \\&=   D_{\lambda}T_{\alpha\sigma\kappa}  
		-
		T_{\rho\sigma\kappa}\tilde{\Gamma}^{\rho}_{\lambda\alpha}
		-T_{\alpha\rho\kappa} \tilde{\Gamma}^{\rho}_{\lambda\sigma}  -T_{\alpha\sigma\rho}\tilde{\Gamma}^{\rho}_{\lambda\kappa}   -T_{\rho\nu\kappa}\Sigma^{\rho\nu}_{\lambda\sigma\alpha} 
		-T_{\rho\sigma\nu}\Sigma^{\rho\nu}_{\lambda\kappa\alpha}-T_{\alpha\rho\nu} \Sigma^{\rho\nu}_{\lambda\kappa\sigma}
		\\&-
	  T_{\rho\nu\kappa} \Sigma^{\rho\nu}_{\lambda\alpha\sigma }   -T_{\rho\sigma\nu} \Sigma^{\rho\nu}_{\lambda\alpha\kappa}      
		-T_{\alpha\rho\nu} \Sigma^{\rho\nu}_{\lambda\sigma\kappa}, 
		\end{align*}
			where we inserted the $\Sigma$-term in the last lines			
				\begin{align*}\Sigma^{\rho\nu}_{\lambda\sigma\kappa}=	 g^{\rho\beta} {\Gamma}^{\nu}_{\tau\kappa}  
		\biggl( S^{\theta\tau}_{\ \ \lambda}D_{\theta}g_{\beta\sigma}   
		+S^{\theta\tau}_{\ \ \sigma}D_{\theta}g_{\beta\lambda} -  	S^{\theta\tau}_{\ \ \beta}
		D_{\theta}g_{\lambda\sigma} 
		\biggr), \\\Sigma^{\rho\nu}_{\lambda\alpha\sigma}=	   g^{\rho\beta} {\Gamma}^{\nu}_{\tau\sigma} 	
				\left( S^{\theta\tau}_{\ \ \lambda}D_{\theta}g_{\beta\alpha}   
				+S^{\theta\tau}_{\ \ \alpha}D_{\theta}g_{\beta\lambda} -  	S^{\theta\tau}_{\ \ \beta}
				D_{\theta}g_{\lambda\alpha} 
				\right),\\\Sigma^{\rho\nu}_{\lambda\alpha\kappa}=	   g^{\rho\beta} {\Gamma}^{\nu}_{\tau\kappa} 	
				\left( S^{\theta\tau}_{\ \ \lambda}D_{\theta}g_{\beta\alpha}   
				+S^{\theta\tau}_{\ \ \alpha}D_{\theta}g_{\beta\lambda} -  	S^{\theta\tau}_{\ \ \beta}
				D_{\theta}g_{\lambda\alpha} 
				\right) .
				\end{align*}
 \end{proof}
 \subsection{Proof of Proposition \ref{prop:riemsym}}\label{proof:riemsym}
\begin{proof}
	The relation we are interested in is the following,  
	\begin{align*}
	[\nabla_{\lambda}, \nabla_{\alpha}]g_{\sigma\kappa}=0.
	\end{align*}
	It is equal to zero by the condition of metricity, however, its l.h.s. reads 
		\begin{align*}
	[\nabla_{\lambda}, \nabla_{\alpha}]  g_{\sigma\kappa} &=\nabla_{\lambda} T_{\alpha\sigma\kappa}-(\alpha\leftrightarrow\lambda)\\&=
	D_{\lambda}T_{\alpha\sigma\kappa}  
	-
	T_{\rho\sigma\kappa}\tilde{\Gamma}^{\rho}_{\lambda\alpha}
	-T_{\alpha\rho\kappa} \tilde{\Gamma}^{\rho}_{\lambda\sigma}  -T_{\alpha\rho\sigma}\tilde{\Gamma}^{\rho}_{\lambda\kappa}   -T_{\rho\nu\kappa}\Sigma^{\rho\nu}_{\lambda\sigma\alpha} 
	-T_{\rho\nu\sigma}\Sigma^{\rho\nu}_{\lambda\kappa\alpha}
	\\&-T_{\alpha\rho\nu} \Sigma^{\rho\nu}_{\lambda\kappa\sigma}-
	T_{\rho\nu\kappa} \Sigma^{\rho\nu}_{\lambda\alpha\sigma }   -T_{\rho\sigma\nu} \Sigma^{\rho\nu}_{\lambda\alpha\kappa}      
	-T_{\alpha\rho\nu} \Sigma^{\rho\nu}_{\lambda\sigma\kappa}-(\alpha\leftrightarrow\lambda)\\&=
	D_{\lambda}T_{\alpha\sigma\kappa}   
	-T_{\alpha\rho\kappa} \tilde{\Gamma}^{\rho}_{\lambda\sigma}  -T_{\alpha\rho\sigma}\tilde{\Gamma}^{\rho}_{\lambda\kappa}   -T_{\rho\nu\kappa}\Sigma^{\rho\nu}_{\lambda\sigma\alpha} 
	-T_{\rho\nu\sigma}\Sigma^{\rho\nu}_{\lambda\kappa\alpha}-T_{\alpha\rho\nu} \Sigma^{\rho\nu}_{\lambda\kappa\sigma}
	\\&      
	-T_{\alpha\rho\nu} \Sigma^{\rho\nu}_{\lambda\sigma\kappa}-(\alpha\leftrightarrow\lambda) ,	\end{align*}
where	we have used the symmetry $T_{\rho\nu\kappa}=T_{\rho\kappa\nu}$. 
	Next, we explicitly calculate the former expression term by term with the exception of the terms containing $\Sigma$. We start with the first term, i.e.	
	\begin{align*}	
	D_{\lambda}T_{\alpha\sigma\kappa}-(\alpha\leftrightarrow\lambda)  & 
	=D_{\lambda}\biggl(D_{\alpha}g_{\sigma\kappa}-g_{\beta\kappa}\overline{\Gamma}^\beta_{\alpha\sigma}-g_{\sigma\beta}\overline{\Gamma}^\beta_{\alpha\kappa} \biggr)-(\alpha\leftrightarrow\lambda) 
	\\&=\biggl( -( D_{\lambda}g_{\beta\kappa})\overline{\Gamma}^\beta_{\alpha\sigma}
	-g_{\beta\kappa} D_{\lambda}\overline{\Gamma}^\beta_{\alpha\sigma}
-( D_{\lambda}g_{\beta\sigma})\overline{\Gamma}^\beta_{\alpha\kappa}
-g_{\beta\sigma} D_{\lambda}\overline{\Gamma}^\beta_{\alpha\kappa} 
\biggr)\\&-(\alpha\leftrightarrow\lambda  ),
		\end{align*}
	 which was for the derivatives of $T$, we also have terms of the following form,	 
		\begin{align*}	
		-T_{\alpha\rho\kappa} \tilde{\Gamma}^{\rho}_{\lambda\sigma} 
	& 
		=-\biggl((D_{ {\alpha}}g_{\rho\kappa})-g_{\beta\kappa}\overline{\Gamma}^\beta_{\alpha\rho}-g_{\rho\beta}\overline{\Gamma}^\beta_{\alpha\kappa} \biggr)\tilde{\Gamma}^{\rho}_{\lambda\sigma} \\&=-\biggl((D_{ {\alpha}}g_{\rho\kappa})\tilde{\Gamma}^{\rho}_{\lambda\sigma}-g_{\beta\kappa}\overline{\Gamma}^\beta_{\alpha\rho}\tilde{\Gamma}^{\rho}_{\lambda\sigma}-g_{\rho\beta}\overline{\Gamma}^\beta_{\alpha\kappa}\tilde{\Gamma}^{\rho}_{\lambda\sigma} \biggr),
			\end{align*}
	similarly, we obtain
		\begin{align*}	
	-T_{\alpha\rho\sigma} \tilde{\Gamma}^{\rho}_{\lambda\kappa} 
	& 
	=-\biggl((D_{ {\alpha}}g_{\rho\sigma})-g_{\beta\sigma}\overline{\Gamma}^\beta_{\alpha\rho}-g_{\rho\beta}\overline{\Gamma}^\beta_{\alpha\sigma}\biggr) \tilde{\Gamma}^{\rho}_{\lambda\kappa}  \\&=-\biggl((D_{ {\alpha}}g_{\rho\sigma})\tilde{\Gamma}^{\rho}_{\lambda\kappa}-g_{\beta\sigma}\overline{\Gamma}^\beta_{\alpha\rho}\tilde{\Gamma}^{\rho}_{\lambda\kappa}-g_{\rho\beta}\overline{\Gamma}^\beta_{\alpha\sigma}\tilde{\Gamma}^{\rho}_{\lambda\kappa} \biggr).
	\end{align*} 
	By reinserting the explicit terms into the commutator we have
	\begin{align*}	
	[\nabla_{\lambda}, \nabla_{\alpha}]  g_{\sigma\kappa}&=
	\biggl( -( D_{\lambda}g_{\beta\kappa})\tilde{\Gamma}^\beta_{\alpha\sigma}
	-g_{\beta\kappa} D_{\lambda}\tilde{\Gamma}^\beta_{\alpha\sigma}
	-( D_{\lambda}g_{\beta\sigma})\tilde{\Gamma}^\beta_{\alpha\kappa}
	-g_{\beta\sigma} D_{\lambda}\tilde{\Gamma}^\beta_{\alpha\kappa}\\&
	-(D_{\lambda}g_{\tau\beta})\Sigma^{\tau\beta}_{\alpha\kappa\sigma}
	-g_{\tau\beta}D_{\lambda}\Sigma^{\tau\beta}_{\alpha\kappa\sigma} 
	-(D_{\lambda}g_{\tau\beta})\Sigma^{\tau\beta}_{\alpha\sigma\kappa}
	-g_{\tau\beta}D_{\lambda}\Sigma^{\tau\beta}_{\alpha\sigma\kappa} 
	\biggr)\\&-\biggl((D_{ {\alpha}}g_{\rho\kappa})\tilde{\Gamma}^{\rho}_{\lambda\sigma}-g_{\beta\kappa}\tilde{\Gamma}^\beta_{\alpha\rho}\tilde{\Gamma}^{\rho}_{\lambda\sigma}-g_{\rho\beta}\tilde{\Gamma}^\beta_{\alpha\kappa}\tilde{\Gamma}^{\rho}_{\lambda\sigma}-  {\Gamma}^{\rho}_{\lambda\sigma}\Sigma_{\alpha\kappa\rho}-  {\Gamma}^{\rho}_{\lambda\sigma}\Sigma_{\alpha\rho\kappa}\biggr)\\&-\biggl((D_{ {\alpha}}g_{\rho\sigma})\tilde{\Gamma}^{\rho}_{\lambda\kappa}-g_{\beta\sigma}\tilde{\Gamma}^\beta_{\alpha\rho}\tilde{\Gamma}^{\rho}_{\lambda\kappa}-g_{\rho\beta}\tilde{\Gamma}^\beta_{\alpha\sigma}\tilde{\Gamma}^{\rho}_{\lambda\kappa}-  {\Gamma}^{\rho}_{\lambda\kappa}\Sigma_{\alpha\sigma\rho}- {\Gamma}^{\rho}_{\lambda\kappa}\Sigma_{\alpha\rho\sigma}\biggr)
	-(\alpha\leftrightarrow\lambda)  \\&
			= 
			-\biggl(   
			 D_{\lambda} \Sigma_{\alpha\kappa\sigma} 
			+ D_{\lambda} \Sigma_{\alpha\sigma\kappa} 
			\biggr)+\biggl(  {\Gamma}^{\rho}_{\lambda\sigma}\Sigma_{\alpha\kappa\rho}+  {\Gamma}^{\rho}_{\lambda\sigma}\Sigma _{\alpha\rho\kappa}\biggr)+\biggl(  {\Gamma}^{\rho}_{\lambda\kappa}\Sigma_{\alpha\sigma\rho}+ {\Gamma}^{\rho}_{\lambda\kappa}\Sigma_{\alpha\rho\sigma}\biggr)  	\\&+g_{\theta\kappa} \tilde{R}^{\theta}_{\;\;\sigma\alpha\lambda}+
			g_{\theta\sigma}\tilde{R}^{\theta}_{\;\;\kappa\alpha\lambda}-(\alpha\leftrightarrow\lambda)\\	=& 
			\biggl(   \nabla_{\alpha} \Sigma_{\lambda\kappa\sigma} +  \nabla_{\alpha} \Sigma_{\lambda\sigma\kappa}
			-\nabla_{\lambda} \Sigma_{\alpha\kappa\sigma} 
			- \nabla_{\lambda} \Sigma_{\alpha\sigma\kappa} 
			\biggr)		+ \tilde{R}_{\kappa\sigma\alpha\lambda}+
		 \tilde{R}_{\sigma\kappa\alpha\lambda}\\	=:&K_{\kappa\sigma\alpha\lambda} + \tilde{R}_{\kappa\sigma\alpha\lambda}+
		 \tilde{R}_{\sigma\kappa\alpha\lambda},
			\end{align*} 	
			where in the last lines we defined
				\begin{align*}	
	K_{\kappa\sigma\alpha\lambda}:&=   \nabla_{\alpha} \Sigma_{\lambda\kappa\sigma} +  \nabla_{\alpha} \Sigma_{\lambda\sigma\kappa}
		-\nabla_{\lambda} \Sigma_{\alpha\kappa\sigma} 
		- \nabla_{\lambda} \Sigma_{\alpha\sigma\kappa}  , 
					\end{align*} 
which renders 
		\begin{align*}	
	\tilde{R}_{\sigma\kappa\alpha\lambda}		=-\tilde{R}_{\kappa\sigma\alpha\lambda}  	-K_{\kappa\sigma\alpha\lambda}.   
						\end{align*} 
	The symmetry for $\overline{R}$ is the same as in the commutative case, the proof may be done in a similar fashion. We start by using $\overline{\Gamma}$ instead of $\tilde{\Gamma}$ using the definition from Theorem \ref{P3.1},		
		\begin{align*}	
	[\nabla_{\lambda}, \nabla_{\alpha}]  g_{\sigma\kappa}&=	
	\biggl( -( D_{\lambda}g_{\beta\kappa})\overline{\Gamma}^\beta_{\alpha\sigma}
	-g_{\beta\kappa} D_{\lambda}\overline{\Gamma}^\beta_{\alpha\sigma}
	-( D_{\lambda}g_{\beta\sigma})\overline{\Gamma}^\beta_{\alpha\kappa}
	-g_{\beta\sigma} D_{\lambda}\overline{\Gamma}^\beta_{\alpha\kappa} 
	\biggr)\\&
	-\biggl((D_{ {\alpha}}g_{\beta\kappa})\tilde{\Gamma}^{\beta}_{\lambda\sigma}-g_{\beta\kappa}\overline{\Gamma}^\beta_{\alpha\rho}\tilde{\Gamma}^{\rho}_{\lambda\sigma}-g_{\rho\beta}\overline{\Gamma}^\beta_{\alpha\kappa}\tilde{\Gamma}^{\rho}_{\lambda\sigma} \biggr)\\&-\biggl((D_{ {\alpha}}g_{\beta\sigma})\tilde{\Gamma}^{\beta}_{\lambda\kappa}-g_{\beta\sigma}\overline{\Gamma}^\beta_{\alpha\rho}\tilde{\Gamma}^{\rho}_{\lambda\kappa}-g_{\rho\beta}\overline{\Gamma}^\beta_{\alpha\sigma}\tilde{\Gamma}^{\rho}_{\lambda\kappa} \biggr)	-(\alpha\leftrightarrow\lambda), 
	\end{align*}
	where we substitute $\tilde{\Gamma}$ yet again,
	\begin{align*}
	[\nabla_{\lambda}, \nabla_{\alpha}]  g_{\sigma\kappa}&=	
	\biggl( -( D_{\lambda}g_{\beta\kappa})\overline{\Gamma}^\beta_{\alpha\sigma}
	-g_{\beta\kappa} D_{\lambda}\overline{\Gamma}^\beta_{\alpha\sigma}
	-( D_{\lambda}g_{\beta\sigma})\overline{\Gamma}^\beta_{\alpha\kappa}
	-g_{\beta\sigma} D_{\lambda}\overline{\Gamma}^\beta_{\alpha\kappa} 
	\biggr)\\&
	-\biggl((D_{ {\alpha}}g_{\beta\kappa})\overline{\Gamma}^{\beta}_{\lambda\sigma}-
	(D_{ {\alpha}}g_{\beta\kappa}) g^{\beta\nu}\Sigma_{\lambda\sigma\nu}
	-g_{\beta\kappa}\overline{\Gamma}^\beta_{\alpha\rho}
	\overline{\Gamma}^{\rho}_{\lambda\sigma}
\\&	+g_{\beta\kappa}g^{\rho\nu} {\Gamma}^\beta_{\alpha\rho}\Sigma_{\lambda\sigma\nu} 
	-g_{\rho\beta}\overline{\Gamma}^\beta_{\alpha\kappa}\overline{\Gamma}^{\rho}_{\lambda\sigma} 
		+g_{\rho\beta} {\Gamma}^\beta_{\alpha\kappa}g^{\rho\nu}\Sigma_{\lambda\sigma\nu}
	\biggr)\\&-\biggl((D_{ {\alpha}}g_{\beta\sigma})\overline{\Gamma}^{\beta}_{\lambda\kappa}-
	(D_{ {\alpha}}g_{\beta\sigma}) g^{\beta\nu}\Sigma_{\lambda\kappa\nu}-g_{\beta\sigma}\overline{\Gamma}^\beta_{\alpha\rho}
	\overline{\Gamma}^{\rho}_{\lambda\kappa}
\\&	+g_{\beta\sigma}g^{\rho\nu} {\Gamma}^\beta_{\alpha\rho}\Sigma_{\lambda\kappa\nu}  
	-g_{\rho\beta}\overline{\Gamma}^\beta_{\alpha\sigma}\overline{\Gamma}^{\rho}_{\lambda\kappa} 
	+g_{\rho\beta} {\Gamma}^\beta_{\alpha\sigma}g^{\rho\nu}\Sigma_{\lambda\kappa\nu}
	\biggr)\\&	-(\alpha\leftrightarrow\lambda ),  
	\end{align*}
	where cancellations   occur on the terms that are symmetric on $\alpha$ and $\lambda$ and we have
	\begin{align*}
	[\nabla_{\lambda}, \nabla_{\alpha}]  g_{\sigma\kappa}&= 
	-\biggl( -
	(D_{ {\alpha}}g_{\beta\kappa}) g^{\beta\nu}\Sigma_{\lambda\sigma\nu} 
		+g_{\rho\kappa} {\Gamma}^\rho_{\alpha\beta}g^{\beta\nu}\Sigma_{\lambda\sigma\nu} 
	+  {\Gamma}^\nu_{\alpha\kappa}\Sigma_{\lambda\sigma\nu}
	\biggr)\\&-\biggl( -
	(D_{ {\alpha}}g_{\beta\sigma}) g^{\beta\nu}\Sigma_{\lambda\kappa\nu} 
		+g_{\rho\sigma} {\Gamma}^\rho_{\alpha\beta}g^{\beta\nu}\Sigma_{\lambda\kappa\nu}  
	+  {\Gamma}^\nu_{\alpha\sigma}\Sigma_{\lambda\kappa\nu}
	\biggr)\\&	-(\alpha\leftrightarrow\lambda ) 	\\&+ \overline{R}_{\kappa\sigma\alpha\lambda}+
	\overline{R}_{\sigma\kappa\alpha\lambda}\\&= \overline{R}_{\kappa\sigma\alpha\lambda}+
	\overline{R}_{\sigma\kappa\alpha\lambda}\\&=0.
	\end{align*} 	 
	By using metricity we conclude 
				\begin{align*}	
	\overline{R}_{\kappa\sigma\alpha\lambda}=-
	\overline{R}_{\sigma\kappa\alpha\lambda}.
\end{align*}

\end{proof}

 	\bibliographystyle{alpha}
 	\bibliography{CLAST1}

\newcommand{\etalchar}[1]{$^{#1}$}
\begin{thebibliography}{ADK{\etalchar{+}}09}

\bibitem[ABD{\etalchar{+}}05]{A3}
Paolo Aschieri, Christian Blohmann, Marija Dimitrijevic, Frank Meyer, Peter
  Schupp, and Julius Wess.
\newblock {A Gravity theory on noncommutative spaces}.
\newblock {\em Class. Quant. Grav.}, 22:3511--3532, 2005.

\bibitem[AC10]{A2}
Paolo Aschieri and Leonardo Castellani.
\newblock {Noncommutative Gravity Solutions}.
\newblock {\em J. Geom. Phys.}, 60:375--393, 2010.

\bibitem[ADK{\etalchar{+}}09]{AS09}
P.~Aschieri, M.~Dimitrijevic, P.~Kulish, F.~Lizzi, and J.~Wess.
\newblock {\em Noncommutative Spacetimes: Symmetries in Noncommutative Geometry
  and Field Theory}.
\newblock Lecture Notes in Physics. Springer Berlin Heidelberg, 2009.

\bibitem[BM14]{BM}
Edwin~J. Beggs and Shahn Majid.
\newblock {Gravity induced from quantum spacetime}.
\newblock {\em Class. Quant. Grav.}, 31:035020, 2014.

\bibitem[CC96]{chamseddine_connes_1996}
Ali~H. Chamseddine and Alain Connes.
\newblock Universal formula for noncommutative geometry actions: Unification of
  gravity and the standard model.
\newblock {\em Physical Review Letters}, 77(24):4868 -- 4871, Sep 1996.

\bibitem[Con95]{C}
A.~Connes.
\newblock {\em Noncommutative Geometry}.
\newblock Elsevier Science, 1995.

\bibitem[CT09]{Connes2}
Alain {Connes} and Paula {Tretkoff}.
\newblock {The Gauss-Bonnet Theorem for the noncommutative two torus}.
\newblock {\em arXiv e-prints}, page arXiv:0910.0188, Oct 2009.

\bibitem[DFR95]{DFR}
Sergio Doplicher, Klaus Fredenhagen, and John~E. Roberts.
\newblock {The Quantum structure of space-time at the Planck scale and quantum
  fields}.
\newblock {\em Commun.Math.Phys.}, 172:187--220, 1995.

\bibitem[DVM96]{DuboisViolette:1995hh}
Michel Dubois-Violette and Thierry Masson.
\newblock {On the first order operators in bimodules}.
\newblock {\em Lett. Math. Phys.}, 37:467--474, 1996.

\bibitem[HR06]{A1}
E.~Harikumar and Victor~O. Rivelles.
\newblock {Noncommutative Gravity}.
\newblock {\em Class. Quant. Grav.}, 23:7551--7560, 2006.

\bibitem[Kos86]{koszul_1986}
J.~L. Koszul.
\newblock {\em Lectures on fibre bundles and differential geometry}.
\newblock Springer-Vlg., 1986.

\bibitem[Lan14]{landi_2014}
Giovanni Landi.
\newblock {\em An Introduction to Noncommutative Spaces and Their Geometries}.
\newblock Springer Berlin, 2014.

\bibitem[Mad99]{madore_1999}
J.~Madore.
\newblock {\em An introduction to noncommutative differential geometry and its
  physical applications}.
\newblock Cambridge University Press, 1999.

\bibitem[Mou95]{Mo}
J~Mourad.
\newblock Linear connections in non-commutative geometry.
\newblock {\em Classical and Quantum Gravity}, 12(4):965, 1995.

\bibitem[MP96]{DV1}
Dubois-Violette M. and Michor P.W.
\newblock {Connnections on central bimodules in noncommutative differential
  geometry}.
\newblock {\em Journal of Geometry and Physics}, 20:218--232, 1996.

\bibitem[MT88]{DV2}
Dubois-Violette M. and Masson T.
\newblock { Derivations and non-commutative differential calculus}.
\newblock {\em Acad. Sci. Paris}, 307:403--408, 1988.

\bibitem[Ros01]{R3}
Marcos Rosenbaum.
\newblock Small-scale structure of space-time and dirac operator.
\newblock {\em International Journal of theoretical Physics}, 40(1), 2001.

\bibitem[SS09]{A4}
Peter Schupp and Sergey Solodukhin.
\newblock {Exact Black Hole Solutions in Noncommutative Gravity}.
\newblock 2009.

\bibitem[Wal84]{Wald:1984rg}
Robert~M. Wald.
\newblock {\em {General Relativity}}.
\newblock Chicago Univ. Pr., Chicago, USA, 1984.

\bibitem[Wes03]{W}
Julius Wess.
\newblock {Deformed coordinate spaces: Derivatives}.
\newblock In {\em {Proceedings, 2nd Southeastern European Workshop on
  Mathematical, theoretical and phenomenological challenges beyond the standard
  model: Perspectives of the Balkan collaborations (BW2003).: Vrnjacka Banja,
  Serbia and Montenegro, August 29-September 3, 2003}}, pages 122--128, 2003.

\end{thebibliography}

\end{document}